\documentclass[onecolumn,draftcls]{IEEEtran}

\usepackage{mleftright}       
\mleftright                   

\usepackage{booktabs,bm,xcolor}         

\usepackage{cite} 

\IEEEoverridecommandlockouts

 \title{Compression with Privacy-Preserving \\Random Access}

 \author{%
 Venkat Chandar\thanks{V.\ Chandar is with DE Shaw, New York, USA. A.\ Tchamkerten is with the Institut Polytechnique de Paris, Palaiseau, France. S.\ Vatedka is with the Indian Institute of Technology Hyderabad, Sangareddy, Telangana, India.},
 Aslan Tchamkerten, Shashank Vatedka
    \thanks{This work was presented in part at the 2024 IEEE International Symposium on Information Theory, Athens, Greece.}\thanks{The work of Shashank Vatedka was supported by a Core Research Grant CRG/2022/004464 from SERB, India.}
 }

\newcommand{\svcomment}[1]{\textcolor{red}{SV: #1}}
\newcommand{\atcomment}[1]{\textcolor{red}{AT: #1}}

\usepackage{basicreq,bbm,xcolor}

\newcommand{\openone}{\leavevmode\hbox{\small1\normalsize\kern-.33em1}}

\newcommand{\ie}{{\textit{i.e.}}}

\newcommand{\boo}[1]{{\mathbbm{#1}}}
\newcommand{\bo}[1]{{\bm{#1}}}

\IEEEaftertitletext{\vspace{-1.5\baselineskip}}
\begin{document}
\maketitle

\begin{abstract}
We show that an i.i.d.\ binary source sequence $X_1,\ldots,X_n$ can be losslessly compressed at any rate above entropy while ensuring that the decoding of any $X_i$ reveals no information about the remaining symbols $\{X_j : j \neq i\}$. 

This problem reduces to a marginal consistency problem induced by the simultaneous privacy and reliability constraints. To address it, we develop a technique based on a geometric representation of codeword distributions, which may be of independent interest.

\end{abstract}

\section{Introduction}\label{sec:introduction}

Several recent works have studied compression with locality properties~\cite{makhdoumi_locally_2015,mazumdar_local_2015,chandar2009locally,tatwawadi_universal_2018,pananjady_effect_2018,vatedka2020loga,vestergaard2021enabling,vatedka_local_2020,kamparaju2022low,vatedka2023local,yuan2024local}, focusing on encoding and decoding primitives, and on algorithms and information-theoretic limits that characterize the minimum number of bits that must be accessed or updated to locally decode or modify a source symbol.

In~\cite{chandar2023data}, we introduced the problem of compression with private local decoding (also referred to as \emph{privacy-preserving random access}): can a binary source sequence  
\begin{equation}
    \bo{X} = (X_1, \ldots, X_n)
\end{equation}
be losslessly compressed into a binary codeword \( C_1, C_2, \ldots, C_{nR} \) such that the local decoding of any \( X_i \) depends only on a subset of codeword symbols \( \bo{C}_{\mathcal{I}_i} \defeq \{ C_j : j \in \mathcal{I}_i \} \), for some \( \mathcal{I}_i \subseteq [nR] \), and reveals no information about the remaining source symbols \( \{ X_j : j \neq i \} \)?

To ensure reliability, the codeword bits \( \bo{C}_{\mathcal{I}_i} \) accessed by local decoder \( i \) must enable a reliable distinction between the hypotheses \( X_i = 0 \) and \( X_i = 1 \).  
To ensure privacy, these accessed bits must be independent of all other source symbols \( \{ X_j : j \neq i \} \).  
The challenge stems from compression; some subsets \( \bo{C}_{\mathcal{I}_i} \) must overlap which potentially compromises local privacy.

Somewhat counterintuitively, we showed that for an i.i.d.\ $\mathrm{Bernoulli}(p)$ source, such a scheme exists and achieves a nontrivial compression rate  
\begin{equation}
R = c \, p \, \log^2\!\left(\tfrac{1}{p}\right),
\end{equation}
for a universal constant \( c > 0 \) and sufficiently small \( p \) (all logarithms are base two).  
However, it remained unclear whether the gap to the entropy  
\begin{equation}
H(p) \defeq -p \log p - (1-p) \log (1-p)
\end{equation}
is intrinsic to the privacy constraint.  
In particular, the scheme exhibits a multiplicative gap of order \( \log (1/p) \) in the small-\( p \) regime.

In this paper, we show that for any i.i.d. $\mathrm{Bernoulli}$ source, private local decoding does not limit compression; there exist schemes that achieve private local decoding at any rate above entropy. The present paper provides the full version of the conference paper~\cite{chandar2024entropy} which only sketched parts of the proofs of some of the results.   

Most recently, \cite{sid2025simple} proposed a conceptually simple, low-complexity scheme that achieves entropy with private local decoding for any i.i.d.\ source, not necessarily binary. We briefly review this scheme and then highlight the differences with the one proposed in this paper.

Given a source sequence $\bo{X}$, we first draw a uniformly random permutation 
$\sigma$ over $[n]$ and form the permuted sequence $\tilde{\bo{X}}$ defined by
$\tilde{X}_{\sigma(i)} = X_i$ for every $i \in [n]$. 
The encoder then outputs
\[
\big(\text{comp}(\tilde{\bo{X}}), \sigma(1),\sigma(2),\ldots,\sigma(n)\big),
\]
where $\text{comp}(\cdot)$ denotes any lossless, entropy-achieving compression
map.

To recover $X_i$, the $i$th local decoder is given
\[
\bo{C}_{\mathcal{I}_i} \defeq 
\big(\text{comp}(\tilde{\bo{X}}), \sigma(i)\big).
\]
It first reconstructs $\tilde{\bo{X}}$ from $\text{comp}(\tilde{\bo{X}})$, then
declares $\tilde{X}_{\sigma(i)}$. 

As stated, the scheme does not fully guarantee privacy since 
$\tilde{\bo{X}}$ reveals the empirical type of $\bo{X}$. This can be remedied by 
appending to $\bo{X}$ a fixed-length suffix of negligible size so as to obtain an extended 
sequence $\bar{\bo{X}}$ of constant composition prior to permutation; the same permute-and-compress scheme is then applied to $\bar{\bo{X}}$. 

A second issue is that the description of $\sigma(1),\ldots,\sigma(n)$ requires on the 
order of $n\log n$ bits, which prevents compression. To overcome this, we amortize the cost of $\sigma$ over several blocks. We first partition $\bo{X}$ into $n/b$ consecutive blocks of size $b=\log n$, then make each block constant composition, and finally use the same random permutation across all blocks. With this modification, the cost of specifying the permutation becomes negligible relatively to $n$, and the overall rate approaches the source entropy as $n$ grows.

In this construction $\sigma$ carries no information and serves as a privacy key. In this sense, the scheme effectively 
decouples privacy from compression. Finally, the approach is inherently asymptotic, as the block length ($\log n$) must be sufficiently large for 
$\text{comp}(\cdot)$ to operate arbitrarily close to the entropy limit 
(see~\cite{sid2025simple} for details).

The scheme developed in this paper addresses compression and privacy jointly, without relying on ``chunking.'' A central technical component is the handling of overlaps between the subsets $\{ \bo{C}_{\mathcal{I}_i} \}$, which may be of independent interest. The objective is to encode a source sequence $\bfx$ into a random $nR$-bit codeword $\bo{C}$ such that each marginal $\bo{C}_{\cI_i}$ is independent of $\{x_j : j \neq i\}$, while still enabling reliable recovery of $x_i$.

A key step is to establish that these marginal constraints can be satisfied simultaneously by a joint codeword distribution---an instance of the classical \emph{marginal problem}~\cite{frechet1951tableaux}. Standard approaches appear insufficient in this setting. We instead develop a technique based on a geometric representation of codeword distributions—the \emph{block marginal polytope}—and use it to prove the existence of a distribution satisfying the required constraints. This perspective may be useful more broadly in related variants of the marginal problem.

\subsection*{Related Works}

Makhdoumi \emph{et al.}~\cite{makhdoumi_onlocallydecsource} introduced locally decodable source coding in information theory.  Mazumdar \emph{et al.}~\cite{mazumdar2015local} gave a fixed-blocklength, entropy-achieving compression scheme enabling efficient local decoding of individual source bits. For a memoryless source $P$, they constructed a rate-$H(P)+\varepsilon$ compressor in which a single source bit can be recovered by probing $\Theta\!\left(\tfrac{1}{\varepsilon}\log\tfrac{1}{\varepsilon}\right)$ bits of the compressed codeword. They also showed that, for non-dyadic sources, any scheme operating at rate $H(P)+\varepsilon$ must probe at least $\Omega(\log(1/\varepsilon))$ bits to achieve vanishing error probability.

Tatwawadi \emph{et al.}~\cite{tatwawadi18isit_universalRA} extended these results to Markov sources and proposed a universal scheme requiring $\Theta\!\left(\tfrac{1}{\varepsilon^2}\log\tfrac{1}{\varepsilon}\right)$ probes per decoded bit.

All of the above works operate in the bit-probe model under fixed-length block coding. Variable-length source coding under local decodability constraints was studied by Pananjady and Courtade~\cite{pananjady2018effect}, who derived upper and lower bounds on the achievable compression rate for sparse sources.

Simultaneous local decoding and re-encoding was studied by Vatedka and Tchamkerten~\cite{vatedka2019local}, who constructed compression schemes for memoryless sources in which the average number of compressed bits accessed or modified per source symbol remains constant, independent of $n$.

A distinct line of research considers locality in the context of channel coding, where each codeword symbol must be recoverable from a small subset of other codeword symbols; see, for instance,~\cite{yekhanin2012locally,gopalan2012locality,tamo2014family,cadambe2015bounds,mazumdar2014update,tamo2016bounds} and references therein.

\subsection*{Paper organization}
Section~\ref{sec:problemstatement} recalls the notion of compression with private local decoding and states the main result, Theorem~\ref{thm:main}. Section~\ref{sec:scheme} presents the compression scheme. The remainder of the paper is devoted to its analysis. Section~\ref{polyrep} develops a representation of codeword distributions that underpins the subsequent arguments. Sections~\ref{pvethm1}, \ref{lapro1}, and~\ref{lapro3} establish the key intermediate results leading to Theorem~\ref{thm:main}. Section~\ref{concrem} concludes the paper.

\section{Main Result}\label{sec:problemstatement}

Throughout the paper, we consider the compression of an i.i.d. Bernoulli($p$) bit string $$\bo{X}=X_1,\ldots,X_n$$ with parameter $p=\Pr (X_1=1)\in (0,1/2]$. We recall the notion of private local decodability \cite{chandar2023data}:

\begin{definition}[Locally decodable code]
Given an integer $n\geq 1$ and a constant $R>0$, a rate-$R$ \emph{locally decodable code} ${\cal{C}}_n$ consists of:
\begin{itemize}
    \item an encoder
which maps each length-n bit string $\bo{x}$ to a random $nR$-bit codeword\footnote{We write ${nR}$ to mean $\lceil nR \rceil$.}
$$\bo{C}=C_1,\ldots,C_{nR}$$ according to a distribution $P_{\bo{C}|\bo{x}}$;
\item a set of $n$ deterministic decoding functions $\{f_i:1\leq i\leq n\}$ where each function $$f_i:\{0,1\}^{|\cI_i|}\to\{0,1\}$$ takes as input a set of codeword components $$\bo{C}_{\cI_i}\defeq \{C_j:j\in\cI_i\}$$ indexed by some set of indices $$\cI_i\subseteq \{1,2,\ldots,nR\},$$ and outputs an estimate $$\hat{X}_i\defeq  f_i(\bo{C}_{\cI_i}).$$
\end{itemize}
\end{definition} 
\begin{definition}[Error probability and privacy]\label{defrelpri}
The error probability of a rate-$R$ locally decodable code ${\cal{C}}_n$ is defined as
    \[
     P_{e}^{(n)} \defeq \Pr (\hat{\bo{X}}\neq \bo{X}),
    \]
    where $$\hat{\bo{X}}\defeq \hat{X}_1,\ldots,\hat{X}_n$$ denotes the concatenation of the estimates of the $n$ local decoders, and where the probability is with respect to the randomness of the source and the randomness of the encoding, that is $$\Pr (\hat{\bo{X}}\neq \bo{X})=\sum_{\bo{x}}P_{\bo{X}}(\bo{x})\sum_{{\bo{c}}:\hat{\bo{x}}(\bo{c}) \ne \bo{x}} P_{\bo{C}|\bo{X}}(\bo{c}|\bo{x}).$$
 The code is said to be \emph{locally privately decodable} (or simply \emph{private}) if 
\begin{align}\label{privatecond}
P_{\bo{C}_{\mathcal{I}_i}|\bo{X}=\bo{x}} = P_{\bo{C}_{\mathcal{I}_i}|X_i=x_i},
\end{align}
for every \( \bo{x} \in \{0,1\}^n \) and every \( i \in [n] \).

\end{definition}
The privacy constraint ~\eqref{privatecond} requires that $\bo{C}_{\cI_i}$ and $\bo{X}_{-i}$ be independent conditioned on $X_i$, almost surely (since $p\notin \{0,1\}$).\footnote{$\bo{X}_{-i} \defeq \{X_j : j \neq i\}$.} Since the source is i.i.d., this conditional independence implies unconditional independence, \emph{i.e.},
\begin{align}\label{altpri}
    P_{\bo{C}_{\cI_i}|\bo{X}_{-i}} = P_{\bo{C}_{\cI_i}},
\end{align}
which can be viewed as an alternative, weaker form of  privacy.

Our main result, stated next, says that private local decodability can be achieved at any rate above entropy:

 \begin{theorem}\label{thm:main}
  For every $0< p \leq \frac{1}{2}$ and $R > H(p)$, there exists a sequence of rate-$R$
compression schemes $\{{\mathcal{C}}_n\}_{n\geq 1}$ that are all locally privately decodable and that achieve vanishing error probability 
  $P_e^{(n)}$ as $n\to \infty$.
\end{theorem}

Theorem~\ref{thm:main} also yields private locally decodable compression
schemes for non-binary i.i.d.\ sources. Let $\bo{X}$ be a memoryless source
over a finite alphabet $\cX$, and represent each symbol $X_i$ by a binary
string $B_i(1),\ldots,B_i(k)$ with
\[
k=\lceil \log_2 |\cX| \rceil .
\]
Denote by $\bo{B}(\ell)$ the length-$n$ sequence formed by the $\ell$th bits
of $X_1,\ldots,X_n$. Since the source is i.i.d., each $B_i(\ell)$ is
Bernoulli$(p_\ell)$ for some $0<p_\ell\le 1/2$.

Encoding the sequences $\bo{B}(1),\ldots,\bo{B}(k)$ independently using the
scheme of Theorem~\ref{thm:main} yields private local recovery of each symbol
$X_i$ (equivalently of its bit representation). The resulting compression rate
can approach any value above
\[
\sum_{\ell=1}^k H(p_\ell),
\]
which may be strictly smaller than $\log |\cX|$. For instance, consider a ternary source distribution $P=(P(a),P(b),P(c))$. Under the binary labeling
\[
a\mapsto 00,\qquad b\mapsto 10,\qquad c\mapsto 01,
\]
we have $k=2$, $p_1=P(b)$, and $p_2=P(c)$. Therefore,
\[
\sum_{\ell=1}^2 H(p_\ell)=H(P(c))+H(P(b))<\log 3
\]
whenever
\[
\min\{P(b),P(c)\}<H^{-1}(\log(3/2))\approx 0.138,
\]
since $H(p)\le 1$ for all $p\in[0,1]$.

\section{Compression scheme}\label{sec:scheme}
\begin{figure}
\begin{center}
  \includegraphics[width=0.25\textwidth]{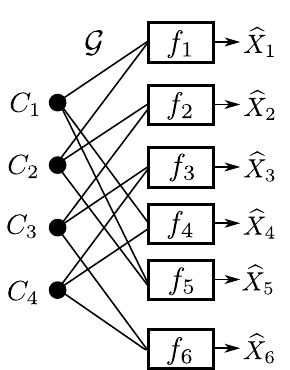}
  \caption{Illustration of the decoder in Section~\ref{sec:scheme} for $n=6$ and $R=2/3$. 
Each local decoder $i$ observes two codeword bits and outputs $\hat{X}_i = f_i(C_{\cI_i})$---for example, $\hat{X}_1 = f_1(C_1,C_2)$.}\label{fig:scheme_new}
\end{center}
\end{figure}
 
 In this section, we describe the coding scheme that establishes Theorem~\ref{thm:main}. 
The encoder maps source sequence \( \bo{X} \) to a random \( nR \)-bit codeword drawn according to a carefully designed conditional distribution \( P_{\bo{C}|\bo{X}} \).

The decoder is defined through a bipartite graph \( \cG \), illustrated in Fig.~\ref{fig:scheme_new}. 
The graph consists of a set \( \cV_L \) of \( nR \) left vertices and a set \( \cV_R \) of \( n \) right vertices. 
The left vertices correspond to the components of the codeword \( \bo{C} \), while each right vertex corresponds to the estimate of a source symbol.
Right vertex \( i \in \cV_R \) is connected to a subset of \( b \) codeword components
\[
\bo{C}_{\mathcal{I}_i} \defeq \{ C_j : j \in \mathcal{I}_i \}, 
\qquad \mathcal{I}_i \subseteq [nR], \; |\mathcal{I}_i| = b,
\]
and is associated with a local decoding function
\[
f_i : \{0,1\}^b \to \{0,1\},
\]
which produces the estimate
\[
\hat{X}_i \defeq f_i(\bo{C}_{\mathcal{I}_i})
\]
for source symbol $X_i$. The global decoding function is obtained by concatenating the local decoders:
\[
f(\bo{C})
\defeq
\big( f_1(\bo{C}_{\mathcal{I}_1}), \ldots, f_n(\bo{C}_{\mathcal{I}_n}) \big)
=
(\hat{X}_1, \ldots, \hat{X}_n)
\defeq
\hat{\bo{X}}.
\]

To prove the existence of a reliable and private code
\[
\mathcal{C}_n = \big(\underbrace{P_{\bo{C}|\bo{X}}}_{\text{encoder}},\underbrace{ \cG, f}_{\text{decoder}} \big),
\]
 we first draw the decoder pair \( (\cG, f) \) at random from a suitable decoder ensemble. 
We then show that, with positive probability, the selected decoder admits an encoding distribution \( P_{\bo{C}|\bo{X}} \) such that the resulting code satisfies the privacy constraint and achieves vanishing error probability \( P_e^{(n)} \).

We next describe the random decoder ensemble.

\subsection{Random decoder ensemble}\label{sct:random}
The decoding graph $\cG$ and the decoding function $f$ are randomly and independently selected as follows. 
For $\cG$, each right vertex is independently connected to a randomly chosen set of $b=\Theta(\log n)$ left vertices (among all $nR$ left vertices) selected uniformly with replacement. 
Further, each $f_i$ is a composition of two functions 
$$f_i \defeq f^{(U)}_i\circ S,$$
where 
$S:\{0,1\}^b\to\{0,1\}^{b'}$ is the syndrome of a Simplex code (dual of a Hamming code), and where each $$f^{(U)}_i:\{0,1\}^{b'}\to\{0,1\}$$ is independently selected by randomly and uniformly sampling a subset of $\{0,1\}^{b'}$ of size $\lceil p 2^{b'}\rceil$ as the preimages of $1$. Accordingly, the parameters $b$ and $b'$ are chosen as
\begin{align}
b = 2^{r}-1\; \text{and}\; b' = 2^r - r - 1\; \text{where}\; r\in \mathbb{N}, r\geq2. \label{implicito}
\end{align}
 In what follows, we parametrize the construction by either $b$ or $n$, depending on the context, and always implicitly assume that \eqref{implicito} holds. We will typically consider the asymptotic $b \to \infty$ with $b = \Theta(\log n)$, which corresponds to $r \to \infty$. In this regime,
\[
\frac{b'}{b} \to 1.
\]

\begin{remark}\label{rem0}
    The reason for $f_i$ to be of the form $f^{(U)}_i\circ S$ (instead of, say, $f^{(U)}_i$) is the following. To achieve compression, the index sets $\mathcal{I}_i$ and $\mathcal{I}_j$ must overlap for some indices $i \neq j$, with $i,j\in [n]$. Consequently, the common bits $\bo{C}_{\mathcal{I}_i \cap \mathcal{I}_j}$ may reveal information about $X_j$ to the $i$-th local decoder (and vice versa), thereby violating privacy. For the decoding graph ensemble defined above, it turns out that $|\cI_i\cap\cI_j|\leq 2$ with high probability (see Lemma~\ref{lemma:smalloverlap}). By defining $f_i$ as $f^{(U)}_i \circ S$ guarantees that 
$f_i(\mathbf{C}_{\mathcal{I}_i})$ cannot be determined from any two coordinates of $\mathbf{C}_{\mathcal{I}_i}$ (see Lemma~\ref{rem:phiideal_consistency}).
     In fact, $S$ could be replaced by the syndrome of the dual of any linear error correcting code with minimum distance $\geq 3$ and with $b'/b\to 1$ as $b\to\infty$.
\end{remark}


To establish Theorem~\ref{thm:main}, we show that, with positive probability, there exists a decoder $(\cG,f)$ that admits an encoding $P_{\bo{C}|\bo{X}}$ that yields a private and reliable code $(P_{\bo{C}|\bo{X}},\cG,f)$.

At the heart of the proof of Theorem~\ref{thm:main}, particularly in constructing a suitable encoding map \( P_{\bo{C}|\bo{X}} \), lies a vector representation of probability mass functions over \( \{0,1\}^{nR} \), which we introduce in the next section.

\section{Block-marginal polytope representation}
\label{polyrep}
\subsection{Definitions and properties}
By Definition~\ref{defrelpri}, the privacy and reliability of $n$ local decoders depend on the marginal distributions on $\{\bo{C}_{\mathcal{I}_i}\}_{i\in [n]}$ induced by the encoder. A block-marginal vector, defined next, represents a set of $n$ such marginals, without reference to an underlying joint distribution.
\begin{definition}[$\cB_\cG$]
Fix a decoding graph $\cG$ with index sets $\cI_1,\ldots,\cI_n \subset [nR]$,
each of size~$b$. The \emph{block-marginal polytope} is
\[
\cB_\cG \defeq 
\Bigl\{
\phi=(\phi^{(1)},\ldots,\phi^{(n)}) \in \mathbb{R}^{n2^b} :
\phi^{(i)}(c^b)\ge 0,\ 
\sum_{c^b\in\{0,1\}^b}\phi^{(i)}(c^b)=1,\ 
\forall i\in[n]
\Bigr\}.
\]

Any vector $\phi\in\cB_\cG$ is called a \emph{block-marginal vector}. It consists
of $n$ blocks
\[
\phi=(\phi^{(1)},\ldots,\phi^{(n)}),
\]
where for each $i\in[n]$ the block $$\phi^{(i)}\in\mathbb{R}^{2^b} $$ is a
probability mass function on $\{0,1\}^b$, referred to as the
\emph{block marginal} associated with the coordinate set $\cI_i$. Coordinates of $\phi^{(i)}$ are indexed by $c^b\in\{0,1\}^b$ in lexicographic order. For instance, for $b=3$ the fifth coordinate of $\phi^{(i)}$ corresponds to $\phi^{(i)}(100)$.

\end{definition}

We now define the set of block-marginal vectors that are realizable as marginals of some joint distribution on the full codeword.

\begin{definition}[$\cP_\cG$]\label{ladefini}
   Fix a decoding graph $\cG$ with index sets $\cI_1,\ldots,\cI_n \subset [nR]$,
each of size~$b$.

Given a codeword $\mathbf c\in\{0,1\}^{nR}$, define the corresponding block-marginal vector $\phi_{\mathbf c}\in\cB_\cG$ by
\[
\phi_{\mathbf c}^{(i)}(c^b)
=
\begin{cases}
1 & \text{if}\;\mathbf c_{\cI_i}=c^b\\
0 & \text{otherwise}
\end{cases}
\qquad i\in[n],\ c^b\in\{0,1\}^b .
\]

Given a codeword distribution $P$ on $\{0,1\}^{nR}$, define the block-marginal vector \emph{generated} by $P$ as 
\begin{align}\label{repr}
\phi_P \defeq \sum_{\mathbf c\in\{0,1\}^{nR}} P(\mathbf c)\,\phi_{\mathbf c}.
\end{align}
Thus, $$\phi_P^{(i)}(c^b)=\Pr_P(\bo{C}_{\cI_i}=c^b).$$ 

We say that block-marginal vector $\phi\in \cB_\cG$ can be \emph{generated} if there exists at least one codeword distribution $P$ such that $$\phi=\phi_P.$$ The set of $\phi$'s that can be generated is denoted by $\cP_\cG$.
\end{definition}

\begin{ex}\label{esemp}
For the graph in Fig.~\ref{fig:phivector_graph_1},\footnote{This graph does not yield compression since $R>1$, and is only used as a toy example to illustrate the block-marginal polytope representation.} the block-marginal vectors for the codewords $(000)$, $(110)$, and $(111)$ have length $2\times 2^{2}=8$, and are given by
\begin{align}
\phi_{000}&=(1,0,0,0,1,0,0,0)&\notag\\
\phi_{110}&=(0,0,0,1,0,0,1,0)&\notag\\
\phi_{111}&=(0,0,0,1,0,0,0,1).&\notag
\end{align}

    If $P$ is the distribution given by $P(000)=0.2$, $P(110)=0.3$, $P(111)=0.5$, then using \eqref{repr} we get
\[
\phi_{P}= (\underbrace{0.2,0,0,0.8}_{\phi_P^{(1)}},\underbrace{0.2,0,0.3,0.5}_{\phi_P^{(2)}}).
\]
\begin{figure}
\begin{center}
  \includegraphics[width=0.25\textwidth]{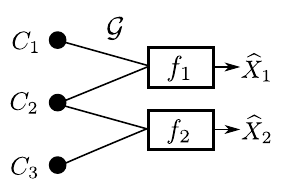}
  \caption{An example to illustrate the block-marginal polytope.}\label{fig:phivector_graph_1}
\end{center}
\end{figure}
\end{ex}

A few remarks are in order:
\begin{enumerate}
\item The vector $\phi_{\mathbf{c}}$, $\bo{c}\in \{0,1\}^{nR}$, depends on the choice of the graph $\cG$ (and similarly for $\phi_P$).
\item  Every $\phi_{\bo{c}}$, $\bo{c}\in \{0,1\}^{nR}$, is a vertex (extreme point) of $\cB_\cG$. However, not every vertex of $\cB_\cG$ is of the form $\phi_{\mathbf{c}}$ for some codeword~$\mathbf{c}$.
\item
If $\phi_1,\phi_2\in \cP_\cG$, then all their convex combinations
$$\phi=\alpha \phi_1 + (1-\alpha)\phi_2 \qquad \alpha \in [0,1]$$ also belong to $\cP_\cG$. Indeed, if $P_1$ and $P_2$ are two codeword distribution such that $$\phi_1=\phi_{P_1}\quad \text{and}\quad\phi_2=\phi_{P_2},$$ then by \eqref{repr}  $$\phi=\alpha \phi_1+(1-\alpha)\phi_2=\phi_{P_\alpha}\in \cP_\cG,$$ where $P_\alpha\defeq\alpha P_1 +(1-\alpha)P_2$.  Hence, the set $\cP_\cG$ is a convex subset of $\cB_\cG$, and its extreme points form a subset of those of $\cB_\cG$.
\item[iv.] The $i$'th block $\phi^{(i)}_P$ of $\phi_P$ corresponds to the distribution of $\bfC_{\cI_i}$, which we interpret as the codeword marginal associated with the $i$'th local decoder. 
\item[v.] Two distinct codeword distributions $P$ and $\tilde{P}$ may produce the same block-marginal vector, \ie, $\phi_P=\phi_{\tilde{P}}$. In contrast, for each $i$, the mapping between a distribution on $\{C_{\cI_i}\}$ and the block marginal $\phi^{(i)}$ is bijective.
\item[vi.] \label{utilee} If $\phi_1\in  \cP_\cG$ and $\alpha \phi_1 + (1-\alpha)\phi_2\in \cP_\cG$ for some $\alpha\in (0,1)$, then this does necessarily imply that  $\phi_2\in \cP_\cG$.
\end{enumerate}

In general, the concatenation of marginal distributions does not correspond to a valid joint distribution. Consequently, not every block marginal vector $\phi$ can be generated.\footnote{For this reason, we use $\phi$ rather than $P$ to denote block-marginal vectors.} For example, consider the graph in Fig.~\ref{fig:phivector_graph_1} and the vector
\[
\phi=(p,0,(1-p),0,0,0,1,0)\in\cB_\cG.
\]
Then
\[
\phi^{(1)}=(p,0,(1-p),0), 
\qquad
\phi^{(2)}=(0,0,1,0).
\]
The marginal $\phi^{(1)}$ assigns probability $p$ to $(C_1,C_2)=(0,0)$ and probability $1-p$ to $(C_1,C_2)=(1,0)$, whereas $\phi^{(2)}$ assigns probability one to $(C_2,C_3)=(1,0)$. These marginals are inconsistent since they assign different probabilities to the shared coordinate $C_2$. Hence no joint distribution on $(C_1,C_2,C_3)$ has marginals $\phi^{(1)}$ and $\phi^{(2)}$.
\begin{definition}[Marginally consistent block-marginal vector] \label{marcon}
Fix a decoding graph $\cG$ with index sets $\cI_1,\ldots,\cI_n \subset [nR]$.  
A block-marginal vector $\phi=(\phi^{(1)},\ldots,\phi^{(n)})$ is said to be \emph{marginally consistent} if for every $i\neq j$ blocks $\phi^{(i)}$ and $\phi^{(j)}$ have their marginals on components $\cI_i\cap\cI_j$ that coincide.
\end{definition}

Note that if $\phi$ and $\tilde{\phi}$ are marginally consistent, then so is
their convex combination $\alpha\phi+(1-\alpha)\tilde{\phi}$, $0\le \alpha\le 1$.
\begin{notation}[Bit and bit-pair consistency]
Given $\phi\in\mathcal P$, we write
\[
\phi^{(i)}_\ell(a)
\;=\;
\sum_{c^b:\,c_\ell=a}\phi^{(i)}(c^b),
\qquad  i\in[n],\ell\in \mathcal I_i,a\in\{0,1\},
\]
for the bit-marginal of $\phi^{(i)}$ at coordinate $\ell$. Similarly, we write
\[
\phi^{(i)}_{\ell,m}(a_1,a_2)
\;=\;
\sum_{c^b:\,c_\ell=a_1,\;c_m=a_2}\phi^{(i)}(c^b),
\qquad i\in[n], \ell,m\in\mathcal I_i,(a_1,a_2)\in\{0,1\}^2,
\]
for the corresponding bit-pair marginal.
\end{notation}
Hence, from Definition~\ref{marcon}, $\phi$ is \emph{bit marginally consistent} if
\[
\phi^{(i)}_\ell=\phi^{(j)}_\ell
\quad\text{for all } i,j\in [n], i\neq j, \text{ and } \ell\in\mathcal I_i\cap\mathcal I_j,
\]
and is \emph{bit-pair marginally consistent} if
\[
\phi^{(i)}_{\ell,m}=\phi^{(j)}_{\ell,m}
\quad\text{for all } i,j\in [n], i\neq j,\text{ and } \ell,m\in\mathcal I_i\cap\mathcal I_j .
\]
If $\phi=(\phi^{(1)},\ldots,\phi^{(n)})$ is bit marginally consistent, then
$\phi^{(i)}_\ell$ is independent of $i$
(whenever defined), and we simply write $\phi_\ell$---and similarly for higher order marginals.

Marginal consistency of $\phi$ does not imply that $\phi$ can be generated, as illustrated next.

\begin{ex}\label{consmarnoprob}
Let $(C_1,C_2,C_3)$ be binary random variables with the following pairwise marginals:\footnote{This would correspond to three block marginal vector, describing the marginals $(C_1,C_2)$, $(C_2,C_3)$, and $(C_1,C_3)$.}
\begin{align*}
P_{C_{1,2}}(00)&=P_{C_{1,2}}(11)=\tfrac12,\\
P_{C_{2,3}}(00)&=P_{C_{2,3}}(11)=\tfrac12,\\
P_{C_{1,3}}(01)&=P_{C_{1,3}}(10)=\tfrac12 .
\end{align*}
These marginals have consistent bit marginals (each variable is Bernoulli$(1/2)$). However, there exists no joint distribution $P_{C_1,C_2,C_3}$ with these pairwise marginals; the first two marginals imply $C_1=C_2=C_3$ with probability one,
whereas the third implies $C_1\neq C_3$ with probability one.
\end{ex}

The following lemma follows from the definition of $\phi$ and the definition of local private decoding (see Definition~\ref{defrelpri}):
\begin{lemma}[Privacy]\label{lempri}
A code is private if and only if, for every $\bo{x}\in\{0,1\}^n$ and $i\in[n]$,
the block-marginal vector $\phi^{(i)}$ depends on $\bo{x}$ only through $x_i$.
\end{lemma}

The next section introduces three block-marginal vectors that will play a central role in establishing Theorem~\ref{thm:main}.
\subsection{The ideal, the uniform, and the approximate block-marginal vectors}
Fix a decoder $(\cG,f)$. We define the three block-marginal vectors and comment thereafter:
\begin{description}
    \item[The ideal block-marginal vector] for a given $\bo{x}\in\{0,1\}^n$, denoted by $\phi_{I,\bo{x}}$, is the vector in $\cB_\cG$ whose $i$-th block is given by
\begin{equation}\label{idblo}
\phi_{I,\bo{x}}^{(i)}(c_{\cI_{i}}) = \begin{cases}
                                    \frac{1}{|f_{i}^{-1}(x_{i})|} & \text{ if } c_{\cI_{i}}\in f_{i}^{-1}(x_{i})\\
                                    0 &\text{ otherwise, }
                     \end{cases}
\end{equation}
for all $i\in [n]$. 


\item[The uniform block-marginal vector,] denoted by $\phi_{{U}}$, is the vector given by $$\phi_{{U}}^{(i)}(c_{\cI_{i}})=2^{-b},\quad i\in [n].$$

\item[The approximate block-marginal vector] 
 (of the ideal block-marginal vector) for a given $\bo{x}\in\{0,1\}^n$ is defined as
\begin{align}\label{phia}
\phi_{A,\bo{x}}\defeq \frac{n}{n+1}\phi_{I,\bo{x}} + \frac{1}{n+1}\phi_{{U}}.
\end{align}
\end{description}

If \( \phi_{I,\bo{x}} \) could be generated (\ie, $\phi_{I,\bo{x}}\in \cP_\cG$), we could encode and decode
$\bo{x}$ privately since \( \phi_{I,\bo{x}}^{(i)} \) depends only on
\( x_i \) (Lemma~\ref{lempri}). Furthermore, because
\(
f_i(c_{\mathcal{I}_i}) = x_i,
\)
the scheme would incur zero error. However, proving that
\(\phi_{I,\bo{x}}\) can be generated with high probability over the
source realizations and the decoder ensemble appears difficult.

In contrast to \( \phi_{I,\bo{x}} \), the vector \( \phi_U \) can always be generated—regardless of \( \mathcal{G} \)—by independently sampling each codeword bit from the Bernoulli\((1/2)\) distribution. Since codewords are chosen independently of \( \bo{x} \), vector \( \phi_U \) yields a scheme that is trivially privately locally decodable but completely unreliable.

Interestingly, the vector \( \phi_{A,\bo{x}} \) can be generated with high probability over both the decoder ensemble and the source realizations (see Proposition~\ref{lemma:phi_NI_generated}), and is ``close'' to
\( \phi_{I,\bo{x}} \). This yields a coding scheme that approximates the ideal block-marginal: it preserves privacy and incurs only a small error probability. Note here that if \( \phi_{A,\bo{x}} \) and \( \phi_U \) can be generated, this does not guarantee that \( \phi_{I,\bo{x}} \) can be generated (see Property~vi. after Example~\ref{esemp}).

Given a decoder $(\cG,f)$, assume that $\phi_{A,\bo{x}}$ can be generated. This means that $\phi_{A,\bo{x}}$ corresponds to some distribution $$P_{\bo{C}|\bo{x}}\defeq P_{A,\bo{x}}$$ over $\{0,1\}^{nR}$. In particular, for any $i\in [n]$ the marginal $P_{A,\bo{x}}^{(i)}$ satisfies  $$P_{A,\bo{x}}^{(i)}=\frac{n}{n+1}P_{I,\bo{x}}^{(i)}+\frac{1}{n+1}P_{U}^{(i)}$$ 
where $P_{I,\bo{x}}^{(i)}$ denotes the uniform distribution over $f_i^{-1}(x_i)$, and where $P_{U}^{(i)}$ denotes the uniform distribution over $\{0,1\}^b$. This corresponds to selecting  $\bo{C}_{\cI_i}$ uniformly at random in the set  $f_i^{-1}(x_i)$, with probability $n/(n+1)$, and selecting a uniformly random bit string $\bo{C}_{\cI_i}$ with probability $1/(n+1)$. Therefore, the local error probability for source realization $\bo{x}$ satisfies
\begin{align}
    \Pr[\hat{X}_i\ne x_i]\defeq \Pr[f_i(\bo{C}_{\cI_i})\ne x_i]\leq 1/(n+1)\quad i\in [n].
\end{align}
Concerning privacy, since neither $P_{I,\bo{x}}^{(i)}$ nor $P_U^{(i)}$ depend on $\bo{x}_{-i}$, any of their convex combination, and $P_{A,\bo{x}}^{(i)}$ in particular, yields a locally private scheme. The following lemma follows:
\begin{lemma}\label{goodcase}
Fix a decoder $(\cG, f)$, a source realization $\bo{x}\in\{0,1\}^n$, and suppose  
$$\phi_{A,\bo{x}}\in \cP_\cG.$$ Then, there exists $P_{\bo{C}|\bo{x}}$ such that the code $(P_{\bo{C}|\bo{x}},\cG, f )$ restricted to $\bo{x}$ is private and achieves error probability 
$$\Pr[\hat{X}_i\ne x_i]\leq \frac{1}{n+1}\qquad i\in [n].$$
\end{lemma}

\section{Proof of Theorem~\ref{thm:main}}\label{pvethm1}
If $p=1/2$, it suffices to let $\hat{\bo{X}}=\bo{C}=\bo{X}$, which gives a rate $R=1$ private and zero-error scheme.

From now on we assume $0<p<1/2$ and consider the decoder ensemble $(\cG,f)$
described in Section~\ref{sct:random}. 
Each local decoder observes a subset $\cI_i$ of $b$ codeword components (determined by $\cG$), computes the length $b'$ syndrome using $S$, and subsequently applies a binary mapping $f^{(U)}_i$.

The parameters are chosen as follows. Let $\varepsilon>0$ be sufficiently
small so that $p+\varepsilon\le 1/2$, and set
\[
R = H(p+\varepsilon), \qquad p_\varepsilon \defeq p+\frac{\varepsilon}{2}.
\]
Let $\varepsilon'\in[0,1/2)$ satisfy
\[
H(\varepsilon') = R - H(p_\varepsilon),
\]
and, as in Section~\ref{sct:random}, let
\[
b=2^{r}-1, \quad b' = 2^r - r - 1,
\]
where we restrict $r$ to be the smallest integer such that 
\[
2^r-1 \geq \frac{16\log n}{\varepsilon'}.
\]
It can be verified that we can choose $b$ satisfying
\[
\frac{16\log n}{\varepsilon'}\leq b\leq \frac{32\log n}{\varepsilon'}.
\]
Notice that $b'/b \to 1$ as $n \to \infty$. 

The proof of Theorem~\ref{thm:main} proceeds in two steps. In the first and main step, presented in Section~\ref{smalldi}, we show that any rate above entropy is achievable with private local decoding and small Hamming distortion; that is, $\bo{X} \oplus \hat{\bo{X}} \leq \delta n$ with high probability for any $\delta > 0$. This is established by showing that the conditions of Lemma~\ref{goodcase} hold with high probability over both the decoder ensemble and the source realizations. In the second step, presented in Section~\ref{zdi}, we strengthen this small-distortion result by demonstrating that the error term $\bo{X} \oplus \hat{\bo{X}}$ can be privately compressed without loss, leveraging a result from~\cite{chandar2023data}.

\subsection{Private compression with small distortion}
\label{smalldi}

The following key proposition that says that $\phi_{A,\bo{x}}\in \cP_\cG$ with high probability over our decoder ensemble, for any source realization $\bo{x} $ in the Hamming ball $$\cB(n, p_\varepsilon)\defeq \{\bo{x}\in \{0,1\}^n: \text{weight}(\bo{x})\leq n p_\varepsilon\}.$$

\begin{proposition}
    \label{lemma:phi_NI_generated}
For any $\bo{x}\in\cB(n, p_\varepsilon)$, we have
\[
\overline{\Pr}[ \phi_{A,\bo{x}}\in \cP_\cG]\geq 1-\frac{\text{poly}({\log n})}{n}
\]
where $\overline{\Pr}[\cdot]$ indicates averaging over the rate $R=H(p+\varepsilon)$ decoder ensemble. 
\end{proposition}
The proof of Proposition~\ref{lemma:phi_NI_generated} is deferred to Section~\ref{lapro1}. 

Now let $(\cG,f)$ be drawn from our  decoder ensemble, and let random variable $\tilde{\bo{X}}$ be independent of $(\cG,f)$ and follow a distribution with support $${{\cal{S}}}_n\subseteq\cB(n, p_\varepsilon).$$ From Proposition~\ref{lemma:phi_NI_generated}, we get
\begin{align}
1-\frac{\text{poly}({\log n})}{n} &\leq \mathbb{E}_{\tilde{\bo{X}}} (\overline{\Pr}[ \phi_{A,\tilde{\bo{X}}}\in \cP_\cG])\nonumber \\
&= \mathbb{E}_{\tilde{\bo{X}}} \mathbb{E}_{\cG,f}(\openone\{\phi_{A,\tilde{\bo{X}}}\in \cP_\cG\}) \nonumber \\
&=\mathbb{E}_{\cG,f} \mathbb{E}_{\tilde{\bo{X}}} (\openone\{\phi_{A,\tilde{\bo{X}}}\in \cP_\cG\})\notag\\
&=\mathbb{E}_{\cG,f} {\Pr}[ \phi_{A,\tilde{\bo{X}}}\in \cP_\cG]
\label{fubi}
\end{align}
where the second equality follows from Fubini's theorem. 
From \eqref{fubi} we deduce the following Corollary:

\begin{corollary}\label{match}
Suppose $\tilde{\bo{X}}$ has a distribution with support ${\cal{S}}_n\subseteq\cB(n, p_\varepsilon)$. Then, there exists a rate $R=H(p+\varepsilon)$ private decoder $(\cG,f)$ such that $${\Pr}[ \phi_{A,\tilde{\bo{X}}}\in \cP_\cG]\geq 1-\frac{\mathrm{poly}({\log n})}{n}.$$
\end{corollary}
Now let $\tilde{\bo{X}}$ be a source with support in ${\cal{S}}_n\subseteq\cB(n, p_\varepsilon)$, let $(\cG,f)$ be the decoder predicted by Corollary~\ref{match}, and let $\tilde{\bo{x}}_0$ be any realization such that 
 $\phi_{A,\tilde{\bo{x}}_0}\in \cP_\cG$.
 
 Given an instance $\tilde{\bo{x}}$ of $\tilde{\bo{X}}$, we distinguish two cases:
\begin{description}
    \item[$i.$ $\phi_{A,\tilde{\bo{x}}}\in \cP_{\cG}$:] 
 From Lemma~\ref{goodcase}, source $\tilde{\bo{x}}$ can be encoded with some distribution $P_{\bo{C}|\tilde{\bo{x}}}$ and decoded with $(\cG,f)$ in a locally private fashion with error probability 
\begin{align}
    \Pr[\hat{\tilde{{X}}}_i\ne {\tilde{{x}}}_i]\leq 1/(n+1)\quad i\in [n],
\end{align}
where $\hat{\tilde{X}}_i$ refers to the $(\cG, f)$-decoder estimate of $\tilde{X}_i$.
\item[$ii.$ $\phi_{A,\tilde{\bo{x}}}\notin \cP_{\cG}$:]
In this case, let us encode $\tilde{\bo{x}}$ with distribution $P_{\bo{C}|\bo{x}_0}$ given by Lemma~\ref{goodcase}---the point here is only to guarantee privacy, not reliability.  
\end{description}
By distinguishing cases $i.$ and  $ii.$ and by using Corollary~\ref{match}, we get for any $i\in [n]$
\begin{align*}
    \Pr[\hat{\tilde{{X}}}_i\ne {\tilde{{X}}}_i]&\leq \Pr[\hat{\tilde{{X}}}_i\ne {\tilde{{X}}}_i|\phi_{A,\tilde{\bo{X}}} \in \cP_\cG]+{\Pr}[ \phi_{A,\tilde{\bo{X}}}\notin \cP_\cG] \\
    &\leq \frac{1}{n+1}+\frac{\text{poly}({\log n})}{n}\\
    &\leq \frac{\text{poly}({\log n})}{n}.
\end{align*}
The following Corollary follows.
\begin{corollary}\label{matcho}
Suppose $\tilde{\bo{X}}$ has a distribution with support ${\cal{S}}_n\subseteq\cB(n, p_\varepsilon)$. Then, there exists a rate $R=H(p+\varepsilon)$ private code $(P_{\bo{C}|\tilde{\bo{X}}},\cG, f )$ such that $${\Pr}[\hat{\tilde{X}}_i\neq \tilde{X}_i]\leq \frac{\text{poly}({\log n})}{n} \quad i\in [n].$$
Hence, by Markov inequality
$${\Pr}[d_H(\tilde{\bo{X}},\hat{\tilde{\bo{X}}})\geq \delta n]\leq \frac{\text{poly}({\log n})}{\delta  n}$$
for any $\delta>0$.\footnote{$d_H(\bo{X},\bo{Y})\defeq \bo{X}\oplus \bo{Y}$ where $\oplus$ denotes the component-wise modulo two sum.}
\end{corollary}
It remains to relate $\tilde{\bo{X}}$ with the source ${\bo{X}}$. Let $\bo{X}$ be i.i.d. Bernoulli($p$), and let $\tilde{\bo{X}}$ have distribution
$$\Pr(\tilde{\bo{X}}=\tilde{\bo{x}})=\Pr({\bo{X}}=\tilde{\bo{x}}|{\bo{X}}\in \cB(n, p_\varepsilon)), \: \tilde{\bo{x}}\in \{0,1\}^n.$$ 
Let $(P_{\bo{C}|\tilde{\bo{X}}},\cG,f)$ be the rate-$H(p+\varepsilon)$ code predicted by Corollary~\ref{matcho} for $\tilde{\bo{X}}$ (with ${\cal S}_n=\cB(n,p_\varepsilon)$), and fix an arbitrary $\bo{x}_0\in\cB(n,p_\varepsilon)$. 
We use this code whenever $\bo{X}\in\cB(n,p_\varepsilon)$; otherwise we encode $\bo{X}$ according to $P_{\bo{C}|\bo{x}_0}$. This gives a private code for source ${\bo{X}}$ and a decoding distortion which we upper bound as:
\begin{align*}
{\Pr}[d_{H}(\bo{X}, \hat{\bo{X}})\geq \delta n]
    &\leq {\Pr}[d_{H}(\bo{X}, \hat{\bo{X}})\geq \delta n|  \bo{X}\in \cB(n, p_\varepsilon)]+{\Pr} [ \bo{X}\notin \cB(n, p_\varepsilon)]\\
    &= {\Pr}[d_H(\tilde{\bo{X}},\hat{\tilde{\bo{X}}})\geq \delta n]+{\Pr} [ \bo{X}\notin \cB(n, p_\varepsilon)]\\
    &\leq \frac{\text{poly}({\log n})}{\delta  n}+2^{-\Theta(n)}\\
    &=\frac{\text{poly}({\log n})}{\delta  n}
\end{align*}
where the second inequality follows from Corollary~\ref{matcho} and standard large deviations arguments.
 The following proposition follows:
\begin{proposition}[Privacy and small average distortion]
        \label{lemma:NI_prob_error}
  Let ${\bo{X}}$ be i.i.d. Bernoulli$(p)$. Then there exists a private rate-$R=H(p+\varepsilon)$ code $( P_{\bo{C}|\bo{X}},\cG,f)$ for the source $\bo{X}$ that  satisfies 
 \begin{align}\label{avg}\Pr[d_{H}(\bo{X},\hat{\bo{X}})\geq \delta n]\leq \frac{\mathrm{poly}({\log n})}{\delta  n}
 \end{align}
  for any $\delta>0$.
   \end{proposition}

Note that the probability on the left-hand side of \eqref{avg} provides only an average guarantee and does not control the maximum distortion incurred by a given source sequence $\bo{x}$. In particular, it does not rule out the possibility that a source sequence $\bo{x}$ is encoded, with non-zero probability, into a codeword $\bo{c}$ such that 
\[
d_H(\bo{x},\hat{\bo{X}}(\bo{c})) = n .
\]
The following stronger result addresses this issue by showing that small distortion can be guaranteed, with overwhelming probability over the source. The proof, deferred to Section~\ref{lapro3}, combines Proposition~\ref{lemma:NI_prob_error} with a coupling argument.

\begin{proposition}[Privacy and small maximum distortion]
    \label{lemma:exp_decay_pe_step4aa}
Let ${\bo{X}}$ be i.i.d. Bernoulli$(p)$. Then there exists a private rate-$R=H(p+\varepsilon)$ code $( P_{\bo{C}|\bo{X}},\cG,f)$ and a set of source sequences $ {\cal{S}}_n\subseteq\{0,1\}^n$ such that
\[
    \Pr[\bo{X}\in {\cal{S}}_n]\geq 1-\frac{ \mathrm{poly}({\log n})}{n},
  \]
and such that, for any $\bo{x}\in{\cal{S}}_n$, $$d_{H}(\bo{x},\hat{\bo{X}})\leq \delta n$$ with probability one---with respect to the random encoding map $P_{\bo{C}|\bo{x}}$.

\end{proposition}

\subsection{From small distortion to zero distortion: Proof of Theorem~\ref{thm:main}}\label{zdi}

By Proposition~\ref{lemma:NI_prob_error} (or Proposition~\ref{lemma:exp_decay_pe_step4aa}), there exists a rate-$R=H(p+\varepsilon)$ private code $\mathcal{C}_1=(P_{\bo{C}(1)|\bo{X}},\cG_1,f_1)$ such that, for the reconstruction $\hat{\bo{X}}=\hat{\bo{X}}(\bo{C}(1))$, the error vector
\[
\bo{Z} \defeq \bo{X}\oplus \hat{\bo{X}}
\]
satisfies
\[
\|\bo{Z}\|_0 \le \delta n
\]
with high probability. The next result shows that $\bo{Z}$ can be re-encoded, with private and zero-error local decoding: 
\begin{theorem}[Theorem 1, \cite{chandar2023data}]\label{elvek}
There exists $\delta_0>0$ such that, for every $\delta \in (0,\delta_0]$, there exists a private zero-error code that compresses any sequence of Hamming weight at most $\delta n$ at rate $O\!\bigl(\delta \log^2(1/\delta)\bigr)$.
\end{theorem}

Assume $\varepsilon,\delta>0$ are sufficiently small constants such that 
$p+\varepsilon \le 1/2$ and $\delta \le \delta_0$, where $\delta_0$ is given in 
Theorem~\ref{elvek}. Let $\mathcal{C}_2=(P_{\bo{C}{(2)}|\bo{Z}},\cG_2,f_2)$ 
be the code predicted by Theorem~\ref{elvek}, and fix an arbitrary vector 
$\bo{z}_0$ with at most $\delta n$ ones. 

If the Hamming weight of $\bo{Z}$ is at most $\delta n$, we use $\mathcal{C}_2$ to privately and losslessly compress $\bo{Z}$ into a rate-$O(\delta\log^2(1/\delta))$ codeword $\bo{C}_{(2)}\sim P_{\bo{C}{(2)}|\bo{Z}}$; otherwise, we encode $\bo{Z}$ 
according to $P_{\bo{C}_{(2)}|\bo{z}_0}$. Using \(\mathcal{C}_1\) and \(\mathcal{C}_2\) yields an overall private code with rate 
\[
H(p+\varepsilon) + O\bigl(\delta \log^{2}(1/\delta)\bigr),
\]
and vanishing error probability. Since \(\varepsilon\) and \(\delta\) can be made arbitrarily small this completes the proof of Theorem~\ref{thm:main}. 
Note that $$Z_i=X_i+\hat{X}_i(\bo{C}_{\cI_i}{(1)})$$ is independent of $\bfX_{-i}$. Accordingly, the successive private decoding of $X_i$ and $Z_i$ reveals no information about~$\bfX_{-i}$.
\qed

\section{Proof of Proposition~\ref{lemma:phi_NI_generated}}

\label{lapro1}




We introduce several definitions and notations used in the subsequent discussion. As before, $(\cG,f)$ denotes a decoder in our ensemble (see Section~\ref{pvethm1}), and $\Pr_{\cG,f}[\cdot]$ (or $\Pr_{\cG}[\cdot]$) denotes probability over the random choice of the decoder.

The following lemma says that the output of the local decoder $i$, namely  
$f_i(\mathbf{C}_{\mathcal{I}_i})$, cannot be determined from any two coordinates of $\mathbf{C}_{\mathcal{I}_i}$.

 \begin{lemma}[Two-bit uninformative overlap]\label{rem:phiideal_consistency}
     Fix a decoder $(\cG,f)$. Then, for any $\bo{x}\in \{0,1\}^{n}$, $i\in [n]$,  $\ell,m\in \cI_i$ with $\ell\ne m$, and $(a_1,a_2)\in \{0,1\}^2$
     \begin{align}
\sum_{c_{\cI_i}: c_\ell=a_1,\,c_m=a_2}
\phi_{I,\bfx}^{(i)}(c_{\cI_i})
=
\frac{1}{4},
\end{align}
equivalently, 
$$\Pr(C_\ell=a_1, C_m=a_2|f_i(C_{\cI_i})=x_i)=\frac{1}{4}.$$

Hence, any bit-pair $(C_\ell, C_m)$ with $\ell,m\in \cI_i$ is independent of $X_i$.\footnote{Recall that by \eqref{idblo} the joint distribution of $(X_i,C_{\cI_i})$ is well-defined, irrespective of whether $\phi_{I,\bo{x}}$ can be generated or not.}
\end{lemma}

\begin{proof}
Let $S$ be a full-rank $b'\times b$ parity-check matrix of the Simplex code.
For $\ell,m\in [b]$ with $\ell\neq m$, consider the linear map
\begin{align*}
F:\{0,1\}^b &\to \{0,1\}^{b'+2} \\
F(c^b) &=\tilde{S}c^b
\end{align*}
where $$ \tilde{S}\defeq \begin{pmatrix}
S \\
e_\ell^\top \\
e_m^\top \\
\end{pmatrix},$$
where $\bfe_\ell$ denotes the $\ell$'th standard ordered basis (column) vector.

Assuming $F$ has rank $b'+2$,
\[
\dim\ker(F) = b - b'-2 ,
\]
and every preimage of $F$ has cardinality $2^{\dim \ker (F)}=2^{b - b'-2}$. Hence, for every syndrome $\mathbf{s}\in\{0,1\}^{b'}$
and every bit pair $(a_1,a_2)\in\{0,1\}^2$,
\[
\bigl|
\{c^b:S c^b=\bo{s},\, c_\ell=a_1,\,c_m=a_2\}
\bigr|
=
2^{b-b'-2},
\]
and therefore
\[
\bigl|
\{c^b:S c^b=\bo{s},\, c_\ell=a_1,\,c_m=a_2\}
\bigr|
=
\frac{1}{4}\bigl|
\{c^b:S c^b=\bo{s}\}
\bigr|.
\]
Since $f^{-1}_i(x)$ can be written as the disjoint
union of syndrome sets
$$f^{-1}_i(x)=\bigcup_{\bo{s}\in \left(f_i^{(U)}\right)^{-1}(x)}\{c^b: S c^b=\bo{s}\},$$
we deduce that
\[
\bigl|
\{c^b\in f^{-1}_i(x): c_\ell=a_1,\,c_m=a_2\}
\bigr|
=
\frac{1}{4}\,|f_i^{-1}(x)|.
\]

Therefore,
\begin{align}
\sum_{c^b: c_\ell=a_1,\,c_m=a_2}
\phi_{I,\bfx}^{(i)}(c^b)
\overset{\eqref{idblo}}{=}\sum_{c^b\in f^{-1}_i(x): c_\ell=a_1,\,c_m=a_2}
\frac{1}{|f_i^{-1}(x)|}=
\frac{1}{4},
\end{align}
and all bit-pair marginals are uniform.

To prove that $F$ has rank $b'+2$, suppose that
\[
\alpha^\top S + \beta_1 e_\ell^\top + \beta_2 e_m^\top = 0,
\]
equivalently, 
\begin{align}
    \label{synd}
\alpha^\top S = \beta_1 e_\ell^\top + \beta_2 e_m^\top,
\end{align}
for some $\alpha\in\{0,1\}^{b'}$ and $\beta_1,\beta_2\in\{0,1\}$. 
The left-hand side belongs to the row space of $S$, \ie, to the dual of the Simplex code. But the dual of the Simplex code is the Hamming code, whose minimum distance is $3$. The right-hand side of \eqref{synd} is a vector with Hamming weight at most $2$. Therefore, there exist no $(\alpha, \beta_1,\beta_2)$ solution to equation \eqref{synd} with $\alpha\ne 0^{b'}$.  Hence the only solution to \eqref{synd} is $\alpha=0^{b'},\beta_1=\beta_2=0 $, thus the rank of $F$ is $b'+2$.
    \end{proof}

\begin{definition}[Eligible perturbations set]\label{df6}
 Fix $\bo{x}\in \{0,1\}^b$ and a decoding graph $\cG$. Define $\cE_{\cG}$ as the set of all vectors $\eta\in\bR^{n2^{b}}$ that can be expressed as $\eta = \phi'-\phi''$ for some marginally consistent $\phi',\phi''\in\cB_\cG$, and that satisfy the following properties:\footnote{Note that $\eta$ does not belong to $\cB_\cG$ since $\sum_{c^b\in \{0,1\}^b}\eta^{(i)}(c^b)=0$, $i\in [n]$.}
 \begin{enumerate}
  \item $|\eta^{{(i)}}_{\ell,m}(a_1,a_2)|\leq n^{-4}$ for all $\ell,m\in\cI_i\cap \cI_j$ for all $i, j$, $i\neq j$, and $(a_1,a_2)\in \{0,1\}^2$
        \item $\eta^{(i)}(c^{b})\leq \frac{1}{n^2 2^{b+1}}$ for all $i\in [n]$ and $c^{b}\in \{0,1\}^{b}$. 
\end{enumerate}
\end{definition}
Proposition~\ref{lemma:phi_NI_generated} will be a direct consequence of the following three lemmas, which we prove thereafter. 

The next lemma shows that, with high probability over the choice of $\cG$, any two local decoders access subsets of codeword symbols whose intersection has size at most two.
\begin{lemma}[Overlap size]\label{lemma:smalloverlap}
We have 
\[
\Pr_{\cG}\Big[|\cI_i\cap\cI_j|\leq 2 \text{ for every }i\neq j\Big]\geq 1-\frac{\mathrm{poly}({\log n})}{n},
\]
(Recall that in our decoder ensemble $|\cI_i|=\Theta(\log n)$).
\end{lemma}

The next lemma says that there exists a ``broad'' class of perturbations of $\phi_{U}$, given by $$\phi_U+\cE_{\cG}\defeq \{\phi_{U}-n\eta:\eta\in\cE_{\cG}\},$$ all of which belong to $\cP_{\cG}$. 
\begin{lemma}\label{lemma:eta_small_existence}
Fix $\bo{x}\in \{0,1\}^b$ and suppose $\cG$ is such that $|\cI_i\cap\cI_j|\leq 2$ for every $i\neq j$.  Then, for every $\eta\in\cE_{\cG}$, we have
  \[
\phi_{U}-n\eta \in\cP_{\cG}.
  \]
\end{lemma}

 The next lemma establishes that $\phi_{I,\bo{x}}$ can be made feasible in
$\cP_{\cG}$ via a perturbation in $\cE_{\cG}$, with high probability.
\begin{lemma}\label{lemma:eta_good_existence}
For any $\bo{x}\in\cB(n,p_\varepsilon)$, 
\[
\Pr_{\cG,f}\Big[\phi_{I,\bo{x}}+\eta\in\cP_{\cG}\text{ for some }\eta\in \cE_{\cG} \Big]\geq 1-\frac{1+o(1)}{n},
\]
as long as 
$$b\geq \frac{16 \log n}{\varepsilon'}. $$
\end{lemma}

\begin{figure}
  \begin{center}
    \includegraphics[width=0.4\textwidth]{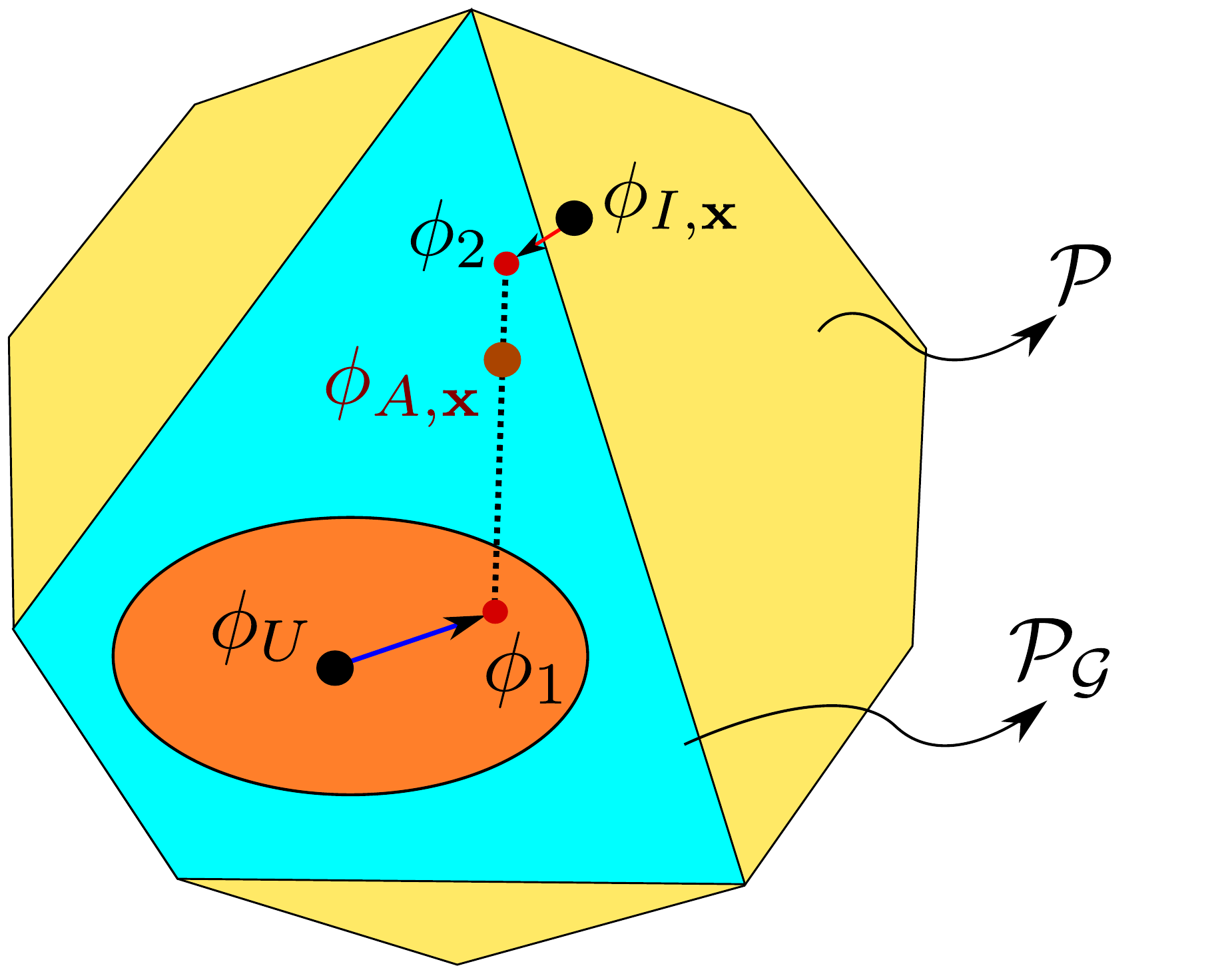}
    \end{center}
\caption{Illustration of the idea behind the proof of Proposition~\ref{lemma:phi_NI_generated}. Recall that $\cP_\cG$ denotes the set of block-marginal vectors that can be generated. This is a convex subset of the set of block-marginal vectors $\cB_\cG$, which is itself convex. The vector $\phi_U$ belongs to $\cP_\cG$. 
Lemma~\ref{lemma:eta_small_existence} says that a ``broad'' class of vectors obtained through perturbations of $\phi_U$ can also be generated. This family can be expressed as $\phi_1\defeq \phi_U-n \eta$ with perturbation $\eta \in \cE_\cG$ and is illustrated by the oval region around $\phi_U$.  
Lemma~\ref{lemma:eta_good_existence} guarantees that with high probability, there exists $\eta \in \cE_\cG$ such that $\phi_2\defeq \phi_{I,\bfx}+\eta$ can be generated.  Consequently, since $\phi_{A,\bfx}$ can be written as a convex combination of $\phi_1$ and $\phi_2$ for some $\eta$, it follows that $\phi_{A,\bfx}$ can also be generated. 
}\label{fig:illustr_phi_ni_gen}
\end{figure}

We now have the ingredients to prove Proposition~\ref{lemma:phi_NI_generated}. 
We need to show that for any $\bo{x}\in\cB(n, p_\varepsilon)$, we have
\[
\overline{\Pr}[ \phi_{A,\bo{x}}\in \cP_\cG]\geq 1-\frac{1}{n^2}(1+o(1)),
\]
where $\overline{\Pr}[\cdot]$ indicates averaging over the rate $R=H(p+\varepsilon)$ decoder ensemble. 

The proof is illustrated in Fig.~\ref{fig:illustr_phi_ni_gen}.
Pick $\bo{x}\in\cB(n,p_\varepsilon)$. By Lemmas~\ref{lemma:smalloverlap}, \ref{lemma:eta_small_existence}, and~\ref{lemma:eta_good_existence}, with probability $$\geq 1-\frac{\text{poly}({\log n})}{n}$$ over our decoder ensemble there exists a decoder $(\cG,f)$ and $\eta\in \cE_{\cG}$ such that 
\begin{align*}
  \phi_{1} &\defeq \phi_{U}-n\eta\quad \text{and}\\
  \phi_{2} &\defeq \phi_{I,\bo{x}}+\eta
\end{align*}
belong to $\cP_{\cG}$.
Any convex combination of such $\phi_{1}$ and $\phi_{2}$ also lies in $\cP_{\cG}$ (see Property iii. before Definition~\ref{marcon}), and in particular
\begin{align*}
  \phi_{A,\bo{x}} &\overset{\text{Eq.\eqref{phia}}}{=} \frac{1}{n+1} \phi_{1} + \frac{n}{n+1} \phi_{2} &\notag\\
           &= \frac{n}{n+1} \phi_{I,\bo{x}} + \frac{1}{n+1}\phi_{U}
\end{align*}
lies in $\cP_{\cG}$. This concludes the proof of Proposition~\ref{lemma:phi_NI_generated}.

\subsection{Proof of Lemma~\ref{lemma:smalloverlap}}
For any distinct pair $(i,j)$, we have
\begin{align}
\Pr_{\cG}[ |\cI_i \cap\cI_j|\geq k] &=\sum_{l=k}^{b}\frac{\nchoosek{b}{l}\nchoosek{nR}{b-l}}{\nchoosek{nR}{b}} = O\left( \frac{b^{2k+1}}{n^k} \right). 
\end{align}
Using the union bound and recalling that $b=\Theta(\log n)$, 
\begin{align*}
    \Pr_{\cG}\Big[ |\cI_i \cap\cI_j|\geq 3, \text{ for some }i\neq j\Big] &\leq \nchoosek{n}{2} O\left( \frac{b^{7}}{n^3} \right) =  \frac{\mathrm{poly}(\log n)}{n} .
\end{align*}

\begin{remark}[Dependence of $(b,b')$ on $n$] 
In our proofs, we choose $b=\Theta(\log n)$. This is to balance the tension between the following two requirements. The proof of Lemma~\ref{lemma:smalloverlap} holds as long as $b=O(\log n)$ (not necessarily $b=\Theta(\log n)$). Later for Lemma~\ref{lemma:eta_good_existence} and~\ref{lemma:exp_valid_codewords} we will require that  $b=\Omega(\log n)$.  
These two requirements force $b$  to grow as $\Theta(\log n)$. 
\end{remark}

\subsection{Proof of Lemma~\ref{lemma:eta_small_existence}: Perturbations of the uniform distribution that can be generated}\label{sec:prf_lemma_perturbation_exist}

Lemma~\ref{lemma:eta_small_existence} can be stated in the following equivalent form.
We are given a collection of $n$ marginal distributions  $\phi^{(1)},\ldots,\phi^{(n)}$, each defined on $\{0,1\}^b$ and all being sufficiently ``close'' to the uniform distribution, and we want to show that there exists a joint distribution $P$ on $\{0,1\}^{nR}$ with marginal $P_{{\cI_i}}$ equal to $\phi^{(i)}$ for every $i$.
This is an instance of the classical \emph{marginal problem}, a family of problems with a long history in the literature; see, \emph{e.g.},~\cite{hoeffding1940masstabinvariante,frechet1951tableaux,csiszar1975divergence,jha2025random}.

A probability distribution satisfying the prescribed marginal constraints exists if and only if there is a nonnegative (column) vector $P \in \mathbb{R}^{2^{nR}}$, representing a probability mass function on $\{0,1\}^{nR}$, solving the linear system
\[
A P = B.
\] The matrix $A$ has dimension $n2^b\times 2^{nR}$ and is given by
\[
A =
\begin{pmatrix}
A^{(1)} \\
A^{(2)} \\
\vdots \\
A^{(n)} 
\end{pmatrix}
\]
where each $2^b \times 2^{nR}$ submatrix $A^{(i)}$, $i\in [n]$, enforces the marginal constraints associated with the index set $\cI_i$. Its rows are indexed by $c^b\in\{0,1\}^b$, and its columns by $\mathbf{c}\in\{0,1\}^{nR}$. The entry in row $c^b$ and column $\mathbf{c}$ equals $1$ if $\mathbf{c}_{\cI_i}=c^b$, and $0$ otherwise. Thus, $A^{(1)},\ldots,A^{(n)}$ are completely determined by $\cG$.

The column vector $B = \phi^T$ consists of the prescribed block marginal probabilities (stacked in the same order as the rows of $A^{(1)},\ldots,A^{(n)}$).

Note that since $$\sum_{\bo{c}} P(\bo{c})=\sum_{c^b}\sum_{\bo{c}:\bo{c}_{\cI}=c^b}P(\bo{c})=\sum_{c^b}\phi^{(i)}(c^b)=1,$$ normalization of $P$ is automatically enforced by each block marginal constraint.

Since for 
\[
B = B_0 \defeq \phi_U^T
\]
the linear system admits the solution $$P^0 = P_U,$$ (the uniform distribution), one might attempt to prove feasibility for perturbed marginals by invoking classical linear-inequality arguments, such as Farkas' lemma. In particular, one could try to show that a nonnegative solution persists when $\phi_U$ is perturbed by an amount of order $n\eta$. 
We briefly explored such approaches. However, because the matrix $A$ is so large, the characterization from Farkas' lemma is difficult to directly analyze and obvious simplifications produce overly restrictive sufficient conditions that do not accommodate the class of perturbations required in our construction.

Instead, we use a constructive approach. Starting from the uniform distribution $P^{(0)}$, we first perturb it to obtain a distribution $P^{(1)}$ whose bit marginals match those of the target block marginals. In the second step, we further perturb $P^{(1)}$ to obtain $P^{(2)}$, which matches all bit-pair marginals. Finally, we perturb $P^{(2)}$ to obtain $P^\star$, which matches the block marginals.
The final step exploits the property $|\cI_i \cap \cI_j| \le 2$: once the bit-pair marginals are fixed, the block marginals can be adjusted independently. Fig.~\ref{fig:phi_existence_proof} illustrates the construction.

\begin{figure}
    \centering
    \includegraphics[width=0.6\linewidth]{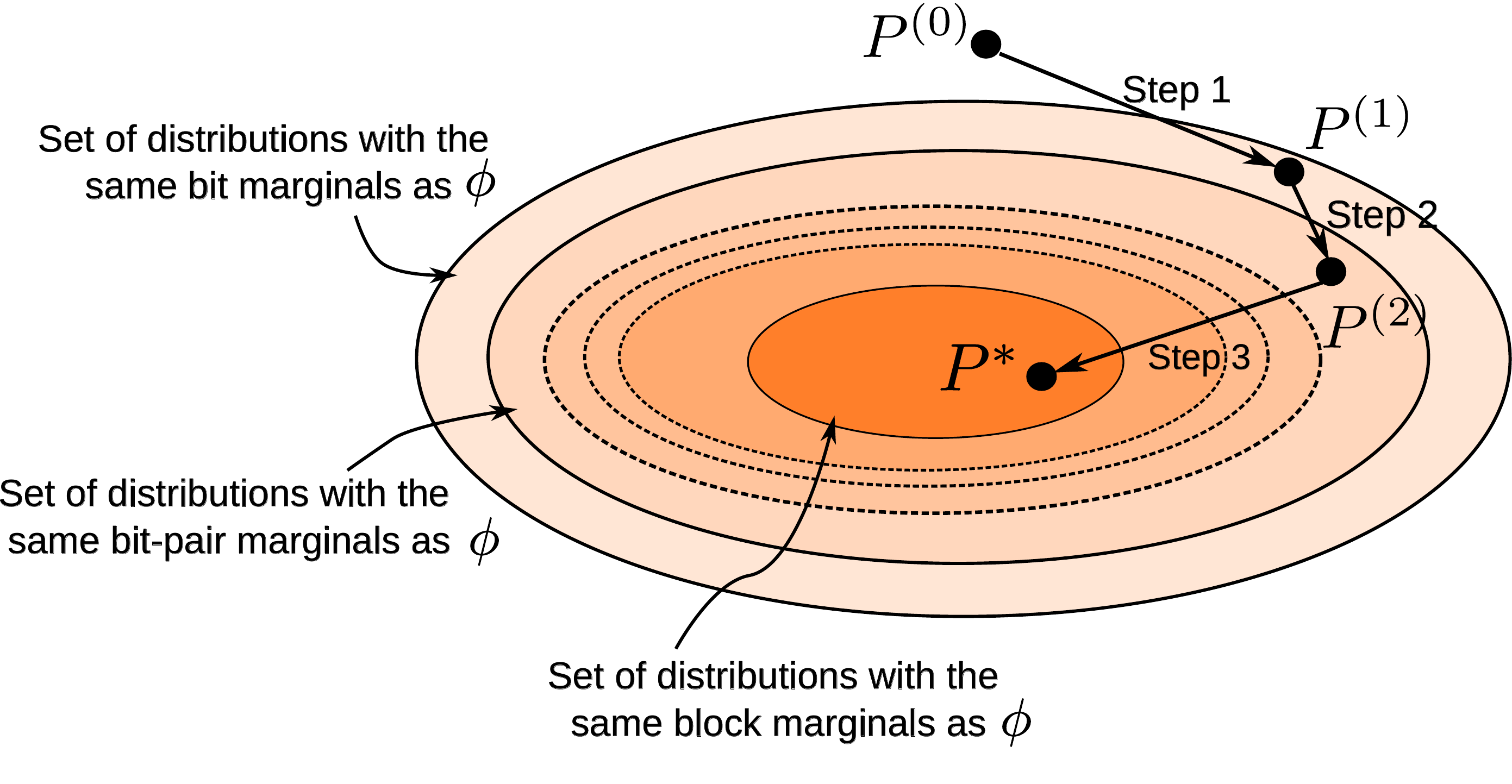}
    \caption{Illustration of the proof of Lemma~\ref{lemma:eta_small_existence}. Each region represents the set of distributions consistent with an increasing number of marginal constraints derived from $\{\phi^{(i)}\}_{i\in [n]}$: the outermost region enforces the bit marginals, the intermediate region enforces the bit-pair marginals, and the innermost region enforces the full block marginals. 
}
    \label{fig:phi_existence_proof}
\end{figure}



\subsubsection*{Proof of Lemma~\ref{lemma:eta_small_existence}}
Let $M \defeq nR$ (assumed to be an integer),  and let
\[
\phi \defeq \phi_U - n\eta = \bigl(\phi^{(1)},\ldots,\phi^{(n)}\bigr)
\]
with $\eta\in \cE_{\cG}$. Note that $\phi^{(i)}$ is a distribution on $\cI_i$ that satisfies 
\[
\phi^{(i)}(c^b)\;\ge\;2^{-b}\Bigl(1-\frac{1}{2n}\Bigr)\qquad \forall\, i\in[n],\ c^b\in\{0,1\}^b,
\]
by Definition~\ref{df6} (Property \emph{ii.}).

Define the set of order-2 marginals
\[
\cM \defeq  \bigl\{(\ell,m):\exists i\in[n]\text{ with }\{\ell,m\}\subseteq \cI_i, \ell<m\bigr\}.
\]
For $\ell\in[M]$ write $\phi_\ell$ for the (well-defined) bit marginal, and for $(\ell,m)\in\cM$
write $\phi_{\ell,m}$ for the corresponding bit-pair marginal.
Set
\[\delta_1 \defeq \max_{\ell\in[M]}\ \max_{a\in\{0,1\}}\Bigl|\phi_{\ell}(a)-\frac12\Bigr|, \qquad
\delta_2 \defeq \max_{(\ell,m)\in\cM}\ \max_{(a_1,a_2)\in\{0,1\}^2}\Bigl|\phi_{\ell,m}(a_1,a_2)-\frac14\Bigr|.\]
By Definition~\ref{df6} (Property \emph{i.}),
\[
\delta_{2}\leq n^{-3},
 \]
and
\begin{align}\label{e1max}
\delta_{1} &\leq \max_{\ell \in [M]}\max_{a_1\in \{0,1\}}\left(\left\vert  \phi_{\ell, m}(a_1,0)-\frac{1}{4}\right\vert + \left\vert \phi_{\ell, m}(a_1,1)-\frac{1}{4}\right\vert  \right)\notag\\
&\leq 2\delta_2\notag\\
&=2n^{-3}.
\end{align}

We construct a global distribution on $\{0,1\}^M$ in stages, starting from the uniform distribution $P_U$. At each stage, we apply perturbations to match a higher order class of marginals of $\phi$: first bits, then pairs, and finally blocks. 

We next introduce the notion of a cylinder lift, along with the two types of perturbations that will be used repeatedly throughout the construction.

\medskip
\noindent\textbf{Cylinder lift.}
For any index set $\cI\subseteq[M]$ and any pmf $P$ on $\{0,1\}^{|\cI|}$, define its \emph{cylinder lift}
$\mathsf{Cyl}_\cI(P)$ on $\{0,1\}^M$ by
\[
\bigl(\mathsf{Cyl}_\cI(P)\bigr)(\bo{c})\;\defeq\;2^{-(M-|\cI|)}\,P(\bo{c}_\cI).
\]
Thus $\mathsf{Cyl}_\cI(P)$ extends $P$ uniformly on the coordinates outside $\cI$.

\medskip
\noindent\textbf{Update identity.}
Let $V,W,W'$ be pmfs on $\{0,1\}^M$ and define $$V' \defeq V +\Delta,$$ where $$\Delta\defeq W - W'$$ is interpreted as a perturbation of $V$ to obtain update $V'$. 
For every $\cI \subseteq [M]$,
\[
V'_\cI = V_\cI + \Delta_\cI,
\]
where $\Delta_\cI \defeq W_\cI - W'_\cI$ is the perturbation restricted to $\cI$.
We will consider two types of local perturbations:
\begin{description}
    \item[Type I.] If $W_\cI = W'_\cI$, then the $\cI$-marginal is preserved, \ie, $V'_\cI = V_\cI$.
    \item[Type II.] If $V_\cI = W'_\cI$, then the $\cI$-marginal is replaced by that of $W$, \ie, $V'_\cI = W_\cI$.
\end{description}

\medskip

Without loss of generality we assume that $\bigcup_{i=1}^n \cI_i = [M]$; any coordinates outside the union
can be made independent $\mathrm{Bernoulli}(1/2)$ without affecting the block marginals on $\{\cI_i\}$.

\noindent\textbf{Stage 0: start from uniform distribution.}
Let $P^{(0)}\defeq P_U$, \ie, $P^{(0)}(\bo{c})=2^{-M}$ for all $\bo{c}\in \{0,1\}^M$.

\medskip
\noindent\textbf{Stage 1: match all bit marginals.}
For each $\ell\in[M]$, define the cylinder lifts
\[
W^\ell \defeq \mathsf{Cyl}_{\{\ell\}}(\phi_\ell),
\qquad
\tilde{W}^\ell \defeq \mathsf{Cyl}_{\{\ell\}}(P^{(0)}_\ell).
\]
Set the update direction $\Delta^\ell \defeq  W^\ell - \tilde{W}^\ell$ and define
\[
P^{(1)} \defeq P^{(0)} + \sum_{\ell'=1}^M \Delta^{\ell'}.
\]
We check that $P^{(1)}$ is a valid distribution on $\{0,1\}^M$ and that its bit marginals correspond to those of $\phi$. 

Each $\Delta^{\ell'}$ has zero total mass, hence $\sum_\bo{c} P^{(1)}(\bo{c})=1$. For the pointwise nonnegativity of $P^{(1)}$, note that
\[
|\Delta^\ell(\bo{c})| \le 2^{-(M-1)} \delta_1,\]
and therefore
\[
P^{(1)}(\bo{c})\ge 2^{-M}\bigl(1-2M\delta_1\bigr)=2^{-\Theta (n )}(1-O(n^{-2})),
\]
since $M=\Theta (n)$ and $\delta_1\leq n^{-3}$. Hence, for all sufficiently large $n$ we have $P^{(1)}(\bo{c})\ge 0$ for all $\bo{c}$.

\emph{Correct bit marginals.}
The $\ell$-marginal of $P^{(1)}$ is given by 
\begin{align}
P^{(1)}&= P^{(0)}_\ell + \sum_{\ell'=1}^M \Delta^{\ell'}_\ell\\
&= (P^{(0)}_\ell + \Delta^\ell_\ell) +\sum_{\ell'\neq \ell}\Delta_\ell^{\ell'}.
\end{align}
All (local) perturbations $\Delta^{\ell'}_\ell$, $\ell'\neq \ell$, are of Type I., hence are equal to zero. Instead, $\Delta^\ell_\ell$ is a Type II. perturbation of $P^{(0)}_\ell$ that results in $\phi_\ell$. Therefore, 
\begin{align*}
P_\ell^{(1)} &=\phi_\ell.
\end{align*}

\medskip
\noindent\textbf{Stage 2: match all bit-pair marginals while preserving bit marginals.}
For each $(\ell,m)\in\cM$, define
\[
W^{{\ell,m}} \defeq \mathsf{Cyl}_{\{\ell,m\}}(\phi_{\ell,m}),
\qquad
\tilde{W}^{{\ell,m}} \defeq \mathsf{Cyl}_{\{\ell,m\}}(P^{(1)}_{\ell,m}),
\]
and $\Delta^{\ell,m} \defeq W^{{\ell,m}} - \tilde{W}^{\ell,m}$.
Define
\[
P^{(2)} \defeq P^{(1)} + \sum_{(\ell',m')\in\cM} \Delta^{\ell',m'}.
\]
Each perturbation $\Delta^{\ell',m'}$ has zero total mass, hence $P^{(2)}$ is normalized. For the nonnegativity of $P^{(2)}$, define 
\begin{align*}
\delta_3 &\defeq \max_{(\ell,m)\in\cM}\ \max_{(a_1,a_2)\in\{0,1\}^2}\bigl|P_{\ell,m}^{(2)}(a_1,a_2)- P_{\ell,m}^{(1)}(a_1,a_2)\bigr|\\
&= \max_{(\ell,m)\in\cM}\ \max_{(a_1,a_2)\in\{0,1\}^2}\bigl|  \sum_{(\ell',m')\in\cM} \Delta^{\ell',m'}_{\ell,m}(a_1,a_2) \bigr|\\
&= \max_{(\ell,m)\in\cM}\ \max_{(a_1,a_2)\in\{0,1\}^2}\bigl|   \Delta^{\ell,m}_{\ell,m}(a_1,a_2) \bigr|,
\end{align*}
where we have used the property that all perturbations $\Delta^{\ell',m'}_{\ell,m}$ with $(\ell',m')\neq (\ell,m)$ are of Type I., hence equal to zero. 
\begin{align*}
\delta_3&= \max_{(\ell,m)\in\cM}\ \max_{(a_1,a_2)\in\{0,1\}^2}\bigl|   \phi_{\ell,m}(a_1,a_2) - P^{(1)}_{\ell,m}(a_1,a_2) \bigr|\\
&\leq \max_{(\ell,m)\in\cM}\ \max_{(a_1,a_2)\in\{0,1\}^2}\bigl|   \phi_{\ell,m}(a_1,a_2) - 1/4\bigr| + \bigl|   P^{(1)}_{\ell,m}(a_1,a_2)-1/4 \bigr| 
\end{align*}
The first term is upper bounded by $\delta_2$.
The second term is
\begin{align*}
\bigl|   P^{(1)}_{\ell,m}-P^{(0)}_{\ell,m} \bigr| &= \bigl| \sum_{\ell'=1}^{nR}\Delta^{\ell'}_{\ell,m} \bigr|= \bigl|\Delta^{\ell}_{\ell,m}+\Delta^m_{\ell,m} \bigr|\\
&\leq \bigl|\Delta^{\ell}_{\ell,m}\bigr|+\bigl|\Delta^m_{\ell,m} \bigr|\\
&\leq \frac{\delta_1}{2} + \frac{\delta_1}{2}=\delta_1,
\end{align*}
where we used the property that $\Delta^{\ell'}_{\ell,m}$ is a Type I perturbation for $\ell'\notin \{\ell,m\}$.
Hence,
\begin{align*}
\delta_3&\leq\delta_2+\delta_1 \\
&=O(n^{-3}).
\end{align*}

Then $|\Delta^{\ell,m}(\bo{c})|\le 2^{-(M-2)}\delta_3$ and therefore
\begin{align}
P^{(2)}(\bo{c})&\ge P^{(1)}(\bo{c}) - |\cM|\,2^{-(M-2)}\delta_3.
\end{align}
Using $|\cM|=O(n\log^2 n)$ and $\delta_1,\delta_3=O(n^{-3})$, $P^{(1)}(c)\ge 2^{-M}(1-2M\delta_1)$, one finds 
\[
P^{(2)}(\bo{c})\;\ge\;2^{-M}\Bigl(1 - 2M\delta_1 - 4|\cM|\delta_3\Bigr) 
                \;=\;2^{-\Theta(n)}\bigl(1 - O(n^{-2}\log^2 n)\bigr),
\]
so for all sufficiently large $n$, $P^{(2)}(\bo{c})\ge0$ for all $\bo{c}$.

\emph{Correct bit-pair marginals.}
The $(\ell,m)$-marginal of $P^{(2)}$ is given by
\begin{align*}
   P^{(2)}_{\ell,m}&= P^{(1)}_{\ell,m} +\sum_{(\ell',m')\in\cM}  \Delta^{\ell',m'}_{\ell, m} \notag \\
   &=P^{(1)}_{\ell,m} + \Delta^{\ell,m}_{\ell,m}+\sum_{(\ell',m')\in\cM, (\ell',m')\neq (\ell,m)} \Delta^{\ell',m'}_{\ell, m}.
\end{align*}
All perturbations $\Delta^{\ell',m'}$ with $(\ell',m')\neq (\ell,m)$ are of Type I., hence equal to zero. Note that this holds even if $(\ell',m')$ and $(\ell,m)$ share a common bit because $\phi$ and $P^{(1)}$ share the same bit marginals (by Stage $1$). Furthermore, $\Delta^{\ell,m}_{\ell, m}$ is a Type II. perturbation of $P^{(1)}_{\ell,m}$ which gives $\phi_{\ell,m}$. Therefore,
$$P^{(2)}_{\ell,m}=\phi_{\ell,m}.$$

Finally, we record a uniform bound on the block-marginal deviation of $P^{(2)}$ from uniform:
for any $\cI\subseteq[M]$ with $|\cI|\geq 2$, 
\begin{align}
\max_{\bo{c}_\cI}\Bigl|P^{(2)}_\cI(\bo{c}_\cI)-2^{-b}\Bigr|
&\;\le\;  \max_{\bo{c}_\cI}\Bigl|P^{(1)}_\cI(\bo{c}_\cI)-2^{-b}\Bigr| + \max_{\bo{c}_\cI}\Bigl|P^{(2)}_\cI(\bo{c}_\cI)-P^{(1)}_{\cI}(\bfc_{\cI})\Bigr| &\notag \\
&\;\le\;  \max_{\bo{c}_\cI}\Bigl|\sum_{\ell=1}^M \Delta^\ell_{\cI}(\bfc_{\cI})\Bigr| + \max_{\bo{c}_\cI}\Bigl| \sum_{(\ell,m)\in\cM} \Delta^{\ell,m}_{\cI}(\bfc_{\cI}) \Bigr| &\notag\\
&\;\le\; 
\frac{M}{2^{|\cI|-1}}\,\delta_1 \;+\; \frac{n\binom{b}{2}}{2^{|\cI|-2}}\,\delta_3
\;\defeq\; \frac{1}{2^{|\cI|}}\delta_{\mathrm{blk}}.\label{eq:deltablk_blk}
\end{align}
Since $\delta_1,\delta_3=O(n^{-3})$ and $b=O(\log n)$, we have $$\delta_{\mathrm{blk}}=O(n^{-2}\log^2n).$$

\medskip
\noindent\textbf{Stage 3: match all block marginals while preserving overlaps.}
For each $i\in[n]$, define
\[
W^{\cI_i} \defeq \mathsf{Cyl}_{\cI_i}(\phi^{(i)}),
\qquad
\tilde{W}^{\cI_i} \defeq \mathsf{Cyl}_{\cI_i}(P^{(2)}_{\cI_i}),
\qquad
\Delta^{\cI_i} \defeq W^{\cI_i} - \tilde{W}^{\cI_i},
\]
and set
\[
P^\star \defeq P^{(2)} + \sum_{i=1}^n \Delta^{\cI_i}.
\]
Each $\Delta^{\cI_i}$ has zero total mass, hence $P^\star$ is normalized. For nonnegativity,
for any $\bo{c}\in\{0,1\}^M$,
\begin{equation}\label{eq:pstarc_1}
P^\star(\bo{c})\ge 2^{-M}- \Bigl| P^{(2)} (\bo{c})- 2^{-M}\Bigr| + \sum_{i=1}^n \Bigl( W^{\cI_i}(\bfc)-2^{-M} -|\tilde{W}^{\cI_i}(\bfc)-2^{-M}|  \Bigr).
\end{equation}
   Using the entrywise bound on $\phi^{(i)}$ and~\eqref{eq:deltablk_blk}, we get 
\begin{align*}
W^{\cI_i}(\bfc)-2^{-M}-|\tilde{W}^{\cI_i}(\bfc)-2^{-M}|
&= \frac{1}{2^{M-b}}\left( \phi^{(i)}(\bfc_{\cI_i})-\frac{1}{2^b} - \Bigl| P^{(2)}_{\cI_i}(\bfc_{\cI_i})-2^{-b} \Bigr| \right)\\
&\geq \frac{1}{2^{M}}\left( -\frac{1}{2n} - \delta_{\mathrm{blk}} \right).
\end{align*}
We can further use~\eqref{eq:deltablk_blk} to bound the second term in~\eqref{eq:pstarc_1}.
Using this and the above in~\eqref{eq:pstarc_1}, we get that for every $\bfc\in\{0,1\}^M$,
\begin{align}
 P^{\star}(\bfc)&\ge \frac{1}{2^{M}}\left(1-\frac{1}{2} - n\,\delta_{\mathrm{blk}}\right)\notag \\
&= \frac{1}{2^{M}}\left(\frac{1}{2}-o(1)\right),
\end{align}
which is nonnegative for all sufficiently large $n$
since $b=\Theta(\log n)$, $M=\Theta(n)$ and $\delta_{\mathrm{blk}}=O(n^{-2}\log^2n)$.

\emph{Correct block marginals.} 
The $\cI_i$ marginal of  $P^{\star}$ is given by 
$$P^{\star}_{\cI_i}=P^{(2)}_{\cI_i} + \sum_{i'=1}^n \Delta^{\cI_{i'}}_{\cI_i}=P^{(2)}_{\cI_i} + \Delta^{\cI_i}_{\cI_i}+\sum_{i'\neq i} \Delta^{\cI_{i'}}_{\cI_i} $$

Perturbations $\Delta^{\cI_{i'}}_{\cI_i}$, $i'\neq i$, are all of Type I. and are equal to zero since $\phi$ and $P^{(2)}$ share the same bit-pair marginals (by Stage $2$) and that $|\cI_i\cap \cI_{i'}|\leq 2$ for all $i'\ne i$ (by Lemma assumption). Finally, $\Delta^{\cI_i}_{\cI_i}$ is a Type II. perturbation of  $P^{(2)}_{\cI_i}$ which gives $\phi_{\cI_i}= \phi^{(i))}$. Hence, 
\begin{align*}
   P^{\star}_{\cI_i}=\phi^{(i)}.
   \end{align*}
We constructed $P^\star$ on $\{0,1\}^M$ with the desired marginals given by $\phi$. This concludes the proof of Lemma~\ref{lemma:eta_small_existence}.

\subsection{Proof of Lemma~\ref{lemma:eta_good_existence}}\label{sec:prf_prob_good_phi}
As a preliminary step, we introduce the notion of \emph{valid codewords}.

\begin{definition}[Valid codewords]
Given a decoder $(\cG,f)$, we say that $\bo{c}$ is a valid codeword for source $\bo{x}$ if $f(\bo{c})=\bo{x}$. The set of all valid codewords for $\bo{x}$ is denoted by $\cC(\bo{x})$.\label{validcdw}
\end{definition}

\begin{remark}\label{rem1}
    Note that for a fixed decoder $(\cG,f)$, if we encode each source sequence $\bo{x}$ by selecting $\bo{C}$ at random with a distribution supported over $\cC(\bo{x})$, we achieve zero error probability. However, depending on $(\cG,f)$ privacy need not be guaranteed.
\end{remark}

To prove Lemma~\ref{lemma:eta_good_existence}, we need to show for any $\bo{x}\in\cB(n,p_\varepsilon)$, 
\[
\Pr_{\cG,f}\Big[\phi_{I,\bo{x}}+\eta\in\cP_{\cG}\text{ for some }\eta\in \cE_{\cG} \Big]\geq 1-\frac{1+o(1)}{n}.
\]
Let $\phi_{V,\bo{x}}$ denote the block-marginal vector induced by the uniform distribution on the valid codewords $\cC(\bo{x})$. We show that $\phi_{V,\bo{x}}$ is close to $\phi_{I,\bo{x}}$, but establishing that $(\phi_{V,\bo{x}} - \phi_{I,\bo{x}}) \in \cE_{\cG}$ with high probability over the decoder appears beyond reach. However, it turns out that the same block-marginal but obtained after a slight expurgation of $\cC(\bo{x})$ yields the desired conclusion.

We now define the expurgated set of codewords through Definitions~\ref{cyli} and~\ref{expuco}, then derive cardinality properties in Lemma~\ref{lemma:exp_valid_codewords} which will be used to establish Lemma~\ref{lemma:eta_good_existence}. 

\begin{definition}[Cylinder set] \label{cyli}Fix $\bo{x}$ and $(\cG,f)$.
For any $i\in [n]$ and $\tilde{c}^{b}\in \{0,1\}^{b}$, we define the cylinder set of valid codewords 
\begin{align}
  \cC_i(\bo{x},\tilde{c}^{b}) &\defeq \Big\{ \bo{c}: \bo{c}\in \cC(\bo{x}),  \bfc_{\cI_{i}}=\tilde{c}^{b}\Big\}.
\end{align}
\end{definition}

\begin{definition}[Expurgated valid codewords]\label{expuco}
Fix $\bo{x}$ and $(\cG,f)$. A cylinder $\cC_i(\bo{x},\tilde{c}^{b})$ is said to be ``bad'' if
\begin{equation}\label{eq:bad_bin}
\frac{|\cC_i(\bo{x},\tilde{c}^{b})|}{|\cC(\bo{x})|} > \frac{1}{|f_{i}^{-1}(x_{i})|} + \frac{1}{n^22^{b}}.
\end{equation}
A codeword $\bo{c}\in\{0,1\}^{nR}$ is said to be bad if it lies in a bad $\cC_i(\bo{x},\tilde{c}^{b})$ for some $(i,\tilde{c}^{b})$, and good otherwise. 

Let $\cC_{\mathrm{exp}}(\bo{x})$ denote the set of codewords obtained by expurgating $\cC(\bo{x})$ of its bad codewords, and let $\phi_{\tilde{V},\bo{x}}$ denote the block-marginal vector induced by the uniform distribution over $\cC_{\mathrm{exp}}(\bo{x})$.

\end{definition}

We now show that with high probability over $(\cG,f)$, every $\bo{x}$ with sufficiently small Hamming weight has exponentially many valid codewords. 

\begin{lemma}\label{lemma:exp_valid_codewords}
  Fix $\bfx$ with Hamming weight $w\leq np_{\varepsilon}$. Let $(\cG,f)$ be randomly chosen from our (rate $R=H(p+\varepsilon)$) decoder ensemble (see Section~\ref{sct:random}). 
 Define
  \begin{align}
    N_{\bo{x}}\defeq \bE|\cC(\bo{x})| &= \exp_{2}\Big( nR+w\log p+(n-w)\log(1-p) +\Theta(n\log(1+2^{-b'})) \Big).&\label{eq:expectation_valid_codewords}
  \end{align}
  Then, for every $\delta,\gamma>0$,
  \begin{align}
    \Pr\left[\big| |\cC(\bo{x})|-N_{\bo{x}}\big|  > \delta N_{\bo{x}} \right] &\leq \frac{1}{\delta^{2}n^{\gamma-1}}, &\label{eq:prob_exp_valid_codewords}
  \end{align}
  as long as 
  \[
    b \geq \left(\frac{\gamma}{H_2^{-1}(R-H_2(p_\varepsilon))}\right)\log n.
  \]
  Furthermore, for every $i\in[n]$ and  $c^b\in\{0,1\}^b$,
  \begin{align}
\bE \left[|\cC(\bfx)|\big\vert \bfc_{\cI_i}=c^b\right]&=\frac{N_{\bfx}}{|f_i^{-1}(x_i)|}, &\label{eq:avg_valid_codewords_condn}
\end{align}
and 
\begin{align}
    \bE\left[\left| |\cC(\bo{x})|- \frac{N_{\bfx}}{|f_i^{-1}(x_i)|}\right|^2  \Bigg\vert
   \bfc_{\cI_i}=c^b \right] &\leq O\left(\frac{N_{\bfx}^2}{n^{\gamma-1}|f_i^{-1}(x_i)|^2}\right).  \label{eq:var_bin_condn}
\end{align}
\end{lemma}
\begin{proof}[Proof of Lemma~\ref{lemma:exp_valid_codewords}]
We prove the lemma using the second moment method. 

Let us fix an arbitrary $\bo{x}\in\{0,1\}^{n}$ with Hamming weight $w\leq np_{\varepsilon}$ and $\bfc\in\{0,1\}^{nR}$.

Note that for every $i\in [n]$ and any fixed $\cG$, we have
\[
\Pr_{f}[f_{i}(\bfc_{\cI_{i}})=1] = \frac{\nchoosek{2^{b'}}{\lceil p2^{b'}\rceil -1}}{\nchoosek{2^{b'}}{\lceil p2^{b'}\rceil}} = \frac{\lceil p2^{b'}\rceil}{2^{b'}}.
\]
Hence,
\begin{align}
  \Pr_{f}[f(\bfc)=\bo{x}] &= \left( \frac{\lceil p2^{b'}\rceil}{2^{b'}}  \right)^{w}\left( \frac{2^{b'}-\lceil p2^{b'}\rceil}{2^{b'}} \right)^{n-w} &\notag\\
                       &= \exp_{2}\Big( w\log p + (n-w)\log (1-p)  + \Theta( n\log(1+2^{-b'})) \Big)\end{align}
as long as $w\leq np_\varepsilon$. It then follows that 
\begin{align}
  \bE|\cC(\bo{x})| &= \bE \left[\sum_{\bfc}1_{\{f(\bfc)=\bo{x}\}}  \right] &\notag\\
   &= \exp_{2}\Big( nR+w\log p+(n-w)\log(1-p) +\Theta(n\log(1+2^{-b'})) \Big)
\end{align}
as long as $w\leq np_\varepsilon$.

For the second moment of $|\cC(\bo{x})|$ we have
\begin{align}
  \bE[|\cC(\bo{x})|^{2}] &= \bE\left[ \sum_{\bfc}\sum_{\tilde{\bfc}} 1_{\{ f(\bfc)=f(\tilde{\bfc})=\bo{x} \}}   \right] &\notag\\
                            &= \bE[|\cC(\bo{c})|] + \sum_{\bfc}\sum_{\tilde{\bfc}:\tilde{\bfc}\neq \bfc} \Pr_{\cG,f}\left[ f(\bfc)=f(\tilde{\bfc})=\bo{x}\right] .  &\label{eq:exp_valid_cw_secondmoment_1}
\end{align}
We separately analyze the probability term for pairs  $(\bfc,\tilde{\bfc})$ in two different regimes, depending on whether the Hamming distance between the two is large or small. Fix an arbitrary $\delta'>0$.

\begin{description}
    \item[Codewords $(\bfc,\tilde{\bfc})$ with $d_{H}(\bfc,\tilde{\bfc})\geq \delta' n$.] 

Fix any $i\in [n]$, and let $\cS\subset [nR]$ denote the set of locations where $\bfc$ and $\tilde{\bfc}$ differ. Then,
\begin{align}
  \Pr_{\cG}[\cI_{i}\cap \cS=\emptyset] &\leq \frac{\nchoosek{n-|\cS|}{b}}{\nchoosek{n}{b}} \leq \left( \frac{n-b}{n} \right)^{|\cS|} &\notag\\
                           &\leq e^{-\delta' b(1+o(1))} &\notag\\
  &\leq O(n^{-\gamma})
\end{align}
if $b\geq \frac{\gamma}{\delta'}\log n$.

If $\cI_{i}\cap \cS\neq \emptyset$, then $\bfc_{\cI_{i}}\neq \tilde{\bfc}_{\cI_{i}}$. In this case, $f_{i}(\bfc_{\cI_{i}})$ and $f_{i}(\tilde{\bfc}_{\cI_{i}})$ are independent. Using the union bound,
\[
\Pr[f_{i}(\bfc_{\cI_{i}})=f_{i}(\tilde{\bfc}_{\cI_{i}})=1] \leq p^{2} + n^{-\gamma}
\]
and
\[
\Pr[f_{i}(\bfc_{\cI_{i}})=f_{i}(\tilde{\bfc}_{\cI_{i}})=0] \leq (1-p)^{2} + n^{-\gamma}.
\]
As the Hamming weight of $\bo{x}$ is $w$,
\begin{align}
  \Pr[f(\bfc)=f(\tilde{\bfc})=\bo{x}] &\leq (p^{2}+n^{-\gamma})^{w}  ((1-p)^{2}+n^{-\gamma})^{n-w}.
\end{align}

Simplifying, we get
\begin{align}
  \Pr[f(\bfc)=f(\tilde{\bfc})=\bo{x}] &\leq \exp_{2}\Big( 2w\log p + 2(n-w)\log(1-p) \Big)\left(1+\frac{2}{p^{2}n^{\gamma}}\right),&\notag
\end{align}
and therefore,
\begin{align}
  \sum_{\bfc}\sum_{\tilde{\bfc}:d_{H}(\bfc,\tilde{\bfc})>\delta' n}\Pr[f(\bfc)=f(\tilde{\bfc})=\bo{x}] 
  &\leq \exp_{2}\Big( 2nR+ 2w\log p + 2(n-w)\log(1-p) \Big) \left(1+\frac{2}{p^{2}n^{\gamma}}\right).&\label{eq:cv_secondmoment_part1}
\end{align}

\item[Codewords $(\bfc,\tilde{\bfc})$ with $d_{H}(\bfc,\tilde{\bfc})< \delta' n$.] 

We have,
\begin{align}
  \sum_{\bfc}\sum_{\tilde{\bfc}:d_{H}(\bfc,\tilde{\bfc})\leq \delta' n}\Pr[f(\bfc)=f(\tilde{\bfc})=\bo{x}] &\leq\sum_{\bfc}\sum_{\tilde{\bfc}:d_{H}(\bfc,\tilde{\bfc})\leq \delta' n}\Pr[f(\bfc)=\bo{x}] &\notag\\
  &\leq 2^{nH_{2}(\delta')(1+o(1))}\bE|\cC(\bo{x})| &\label{eq:cv_secondmoment_part2}
\end{align}
which is asymptotically much smaller than (\ref{eq:cv_secondmoment_part1}) if we choose $\delta'$ such that $$H_{2}(\delta')<(R+\frac{w}{n}\log p+(1-\frac{w}{n})\log(1-p))/2.$$
\end{description}
We now have the ingredients to upper-bound the variance.

Combining (\ref{eq:cv_secondmoment_part1}) and (\ref{eq:cv_secondmoment_part2}), for $n$ large enough we get
\begin{align}
  \mathrm{Var}(|\cC(\bo{x})|)&\leq \exp_{2}\Big(  2nR+2w\log p+2(n-w)\log(1-p) \Big)\frac{4(1+o(1))}{pn^{\gamma}}
\end{align}

Using~\eqref{eq:expectation_valid_codewords} and the above with Chebyshev's inequality gives us~\eqref{eq:prob_exp_valid_codewords},
\begin{align*}
\Pr\left[\big| |\cC(\bo{x})|-N_{\bo{x}}\big|  > \delta N_{\bo{x}} \right] &\leq \frac{1}{\delta^{2}n^{\gamma - 1}}. 
\end{align*}

Using similar calculations, \emph{mutatis mutandis}, we can show that 
for any $\bo{x}$ with Hamming weight $w\leq np_{\varepsilon}$, and any fixed $i\in[n]$, $\cI\in \nchoosek{[n]}{2^b}$, $f_i$, $c^b\in\{0,1\}^b$,
  \begin{align}
    \bE\left[|\cC(\bo{x})|\big\vert \cI_i,f_i,\bfc_{\cI_i}=c^b \right] &= N_{\bfx}(i) &\label{eq:expect_valid_condn}  
   \end{align}
and
   \begin{align}
    \mathrm{Var}\left[|\cC(\bo{x})|\big\vert \cI_i,f_i,\bfc_{\cI_i}=c^b \right] &= \frac{1}{n^{\gamma-1}}N^2_{\bfx}(i)(1+o(1))  &\label{eq:var_valid_condn}
   \end{align}
  where
   \begin{align}
    N_{\bfx}(i) &\defeq \exp_{2}\Big( nR+w\log p+(n-w)\log(1-p) -b -x_i\log p-(1-x_i)\log(1-p)&\notag\\
                        &\qquad \qquad +\Theta(n\log(1+2^{-b'})) \Big).
  \end{align}

The expectations above are taken with respect to $f_1,\ldots,f_{i-1},f_{i+1},\ldots,f_{n}$ and $\cI_1,\ldots,\cI_{i-1},\cI_{i+1},\ldots,\cI_n$.

The second moment method yields, for every $\delta>0$,
  \[
    \Pr\left[\big| |\cC(\bo{x})|- N_{\bfx}(i) \big|> \delta N_{\bfx}(i) \Big\vert \cI_i,f_i,\bfc_{\cI_i}=c^b \right] \leq \frac{1}{\delta^{2}n^{\gamma-1}}.
  \]
As $f_i,\cI_i$ are independent of $\{f_j,\cI_j:j\in [n]\backslash\{i\}\}$, we can average~\eqref{eq:expect_valid_condn} over $f_i$ and then $\cI_i$ to get
\[
\bE \left[|\cC(\bfx)|\big\vert \bfc_{\cI_i}=c^b\right]=\frac{N_{\bfx}}{|f_i^{-1}(x_i)|}.
\]
Note that $|f_i^{-1}(x_i)|$ depends only on $(p,b',x_i)$ and is independent of the realization of $f_i$. Using the upper bound on the conditional variance~\eqref{eq:var_valid_condn}, we get 
\begin{align}
    \bE\left[\left| |\cC(\bo{x})|- \frac{N_{\bfx}}{|f_i^{-1}(x_i)|}\right|^2  \Bigg\vert \bfc_{\cI_i}=c^b \right] &\leq O\left(\frac{N_{\bfx}^2}{|f_i^{-1}(x_i)|^2n^{\gamma -1}}\right), & \label{eq:var_bin_condn}
\end{align}
which completes the proof.
\end{proof}

We now have the ingredients to prove Lemma~\ref{lemma:eta_good_existence}.
\subsubsection*{Proof  of Lemma~\ref{lemma:eta_good_existence}}
Recall the definition of $\phi_{I,\bfx}^{(i)}$ and $\phi_{V,\bfx}^{{(i)}}$
\begin{align*}
\phi_{I,\bfx}^{(i)}(c^b)&=\begin{cases}
  \frac{1}{|f_i^{-1}(x_i)|} &\text{ if }f_i(c^b)= x_i\\
  0 &\text{ otherwise}. 
\end{cases}&\notag\\
\phi_{V,\bfx}^{(i)}(c^b)&=
\frac{|\cC_{i,\bfx}(c^{b})|}{|\cC(\bfx)|}
\end{align*}
for $i\in[n]$ and $c^b\in\{0,1\}^b$. Let us define
\[
\phi_{\tilde{V},\bfx}^{(i)}(c^b) \defeq \frac{|\cC_{exp,i}(\bfx,c^b)|}{|\cC_{exp}(\bfx)|},
\]
which is the block marginal corresponding to the distribution where the codewords are chosen uniformly at random from $\cC_{exp}(\bfx)$.
We now show that for any $\bo{x}$ with Hamming weight at most $np_\varepsilon$,
  \[
    \Pr_{\cG,f}[\phi_{\tilde{V},\bo{x}}-\phi_{I,\bo{x}}\notin \cE_{\cG}]\leq O\left(\frac{1}{n}\right).
  \]
  This will imply Lemma~\ref{lemma:eta_good_existence} with $\eta=\phi_{\tilde{V},\bo{x}}-\phi_{I,\bo{x}}$. 
  
  From the definition of $\cC_{exp}(\bfx)$, we are guaranteed that for all $i\in[n]$,
  \[
    \phi^{(i)}_{\tilde{V},\bfx}(c^b) \leq \frac{1}{|f_i^{-1}(x_i)|}+\frac{1}{n^2 2^b}.
  \]
  It suffices to show that with high probability, all the bit-pair marginals of $\phi_{\tilde{V},\bo{x}}$ are  close to $1/4$.  
  
  By Lemma~\ref{lemma:exp_valid_codewords}, a typical $\bfx$ has exponentially many valid codewords  and the first step will be to show that the expurgation process (with high probability) only removes a vanishing fraction of these. We use this to show that the bit-pair marginals of $\phi_{V,\bfx}$ and $\phi_{\tilde{V},\bfx}$ are close.  Finally, we use Lemma~\ref{lemma:exp_valid_codewords} to show that the bit-pair marginals of $\phi_{V,\bfx}$ are close to uniform and the result follows with a straightforward application of the union bound. 

We first show that with high probability, the expurgation process does not remove too many codewords.
From Lemma~\ref{lemma:exp_valid_codewords}, we have
\begin{align}
    \bE\left[\left| |\cC_i(\bo{x},\tilde{c}^b)|- \frac{|\cC(\bfx)|}{|f_i^{-1}(x_i)|}\right|^2   \right] &\leq O\left(\frac{N_{\bfx}^{2}}{n^{\gamma-1}|f_i^{-1}(x_i)|^2}\right), & \label{eq:var_bin_condn}
\end{align}
for 
\[
    b \geq \left(\frac{\gamma}{H_2^{-1}(R-H_2(p_\varepsilon))}\right)\log n.
  \]
Summing over $\tilde{c}^{b}$, and observing that $|f_{i}^{-1}(x_{i})|=\Theta(2^{b})$,
\begin{align}
    \bE\left[\sum_{\tilde{c}^{b}}\left| |\cC_i(\bo{x},\tilde{c}^b)|- \frac{|\cC(\bfx)|}{|f_i^{-1}(x_i)|}\right|^2   \right] &\leq O\left(\frac{N_{\bfx}^{2}}{n^{\gamma-1}2^{b}}\right). &\notag
\end{align}
Invoking Lemma~\ref{lemma:exp_valid_codewords} and Markov's inequality,
\begin{align}
\Pr\left[\sum_{\tilde{c}^{b}}\left| |\cC_i(\bo{x},\tilde{c}^b)|- \frac{|\cC(\bfx)|}{|f_i^{-1}(x_i)|}\right|^2 \geq \frac{2N_{\bfx}^{2}}{n^{\gamma -3}|f_{i}^{-1}(x_{i})|}   \right] &\leq O\left(\frac{1}{n^{2}}\right). &\label{eq:bd_1}
\end{align}
Also from Lemma~\ref{lemma:exp_valid_codewords}, we have
\[
\Pr\bigl[|\cC(\bfx)|\geq \sqrt{2} N_{\bfx}\bigr]\leq O\left(\frac{1}{n^{\gamma-1}}\right),
\]
which combined with~\eqref{eq:bd_1} gives
\begin{align*}
\Pr\left[\sum_{\tilde{c}^{b}}\left| |\cC_i(\bo{x},\tilde{c}^b)|- \frac{|\cC(\bfx)|}{|f_i^{-1}(x_i)|}\right|^2 \geq \frac{|\cC(\bfx)|^{2}}{n^{\gamma -3}|f_{i}^{-1}(x_{i})|}   \right] &\leq O\left(\frac{1}{n^{2}}\right).
\end{align*}
Rearranging the above and using the definition of $\phi_{V,\bfx}$, we have for every $i\in[n]$,
\begin{align}
  \Pr\left[\sum_{c^b\in\{0,1\}^b} \left|\phi_{V,\bfx}^{(i)}(\tilde{c}^b) - \frac{1}{|f_i^{-1}(x_i)|} \right|^2\geq \frac{1}{n^{\gamma-3}|f_i^{-1}(x_i)|} \right]&\leq O\left(\frac{1}{n^{2}}\right) .&\label{eq:l2_bound_1}
\end{align}

Using Lemma~\ref{lemma:excess_fraction_bound} (see Appendix) with $\alpha = n^{\gamma-3}$ and $\varepsilon = n^{-2}$ in~\eqref{eq:l2_bound_1}, we obtain that, for every $i \in [n]$, the probability that too many codewords must be expurgated from bad cylinders is small, \emph{i.e.},
\[
\Pr\!\left[
\sum_{c^b \in \{0,1\}^b}
\left[
\phi_{V,\bfx}^{(i)}(c^b)
    - \left(1 + \frac{1}{n^2}\right)
      \frac{1}{|f_i^{-1}(x_i)|}
\right]^+
\;\ge\;
\frac{1}{n^{\gamma-7}}
\right]
\;\le\;
O\left(\frac{1}{n^{2}}\right).
\]
or equivalently,
\[
\Pr\!\left[
\sum_{c^b \in \{0,1\}^b}
\left[
|\cC_i(\bfx,\tilde{c}^b)|
    - \left(1 + \frac{1}{n^2}\right)
      \frac{|\cC(\bfx)|}{|f_i^{-1}(x_i)|}
\right]^+
\;\ge\;
\frac{1}{n^{\gamma-7}}|\cC(\bfx)|
\right]
\;\le\;
O\left(\frac{1}{n^{2}}\right).
\]
Taking a union bound,
\[
\Pr\!\left[
\sum_{c^b \in \{0,1\}^b}
\underbrace{\left[
|\cC_i(\bfx,\tilde{c}^b)|
    - \left(1 + \frac{1}{n^2}\right)
     \frac{|\cC(\bfx)|}{|f_i^{-1}(x_i)|}
\right]^+}_{\text{Number of excess codewords in cylinder set }}
\;\ge\;
\frac{1}{n^{\gamma-7}}|\cC(\bfx)|
\text{ for some }i\in[n]\right]
\;\le\;
O\left(\frac{1}{n}\right).
\]
We therefore have a bound on the total number of codewords that need to be expurgated (with high probability). 
\begin{align}
\Pr\Big[|\cC(\bfx)|-|\cC_{exp}(\bfx)|>n^{-\gamma+7}|\cC(\bfx)|\Big]&\leq O\left(\frac{1}{n}\right).&\label{eq:pr_exp_closeto_valid}
\end{align}
We now use this to show that the bit-pair marginals of $\phi_{\tilde{V},\bfx}$ are close to those corresponding to $\phi_{V,\bfx}$.

  Recall that
  \[
    \phi_{\tilde{V},\bfx,\ell,m}(a_1,a_2) = \frac{\sum_{\tilde{c}^b:c_{\ell}=a_1,c_m=a_2}|\cC_{exp,i}(\bfx,\tilde{c}^b)|}{|\cC_{exp}(\bfx)|}. 
  \]
  Since the expurgation process only deletes codewords,~\eqref{eq:pr_exp_closeto_valid} implies that with probability $1-O(1/n)$,
  \[
    \phi_{V,\bfx,\ell,m}(a_1,a_2) \left(\frac{1}{1-n^{-\gamma+7}}\right) \geq \phi_{\tilde{V},\bfx,\ell,m}(a_1,a_2) \geq \phi_{V,\bfx,\ell,m}(a_1,a_2) (1-n^{-\gamma+7}),
  \]
  for all $(a_1,a_2)\in\{0,1\}^2$. 
Therefore, we have
  \begin{align}
   \Pr\left[  \Big|\phi_{\tilde{V},\bfx,\ell,m}(a_1,a_2) - \phi_{V,\bfx,\ell,m}(a_1,a_2)\Big| >  \frac{1}{n^{(\gamma-7)/2}} \phi_{V,\bfx,\ell,m}(a_1,a_2) \right] &\leq O\left(\frac{1}{n}\right), &\label{eq:bound_phiv_phivtilde}
  \end{align}
  where in the last step, we have used the property that $1/(1-\epsilon)\leq 1+\sqrt{\epsilon}$ for  $\epsilon<0.3$.

Using Lemma~\ref{lemma:exp_valid_codewords}, we get concentration bounds on the bit-pair marginals corresponding to $\phi_{V,\bfx}$ by averaging over the $\tilde{c}^b$.
 The bit-pair marginals satisfy
  \begin{align*}
    \Pr\left[ \left|\phi_{V,\bfx,\ell,m}(a_1,a_2)-\frac{1}{4} \right|>\frac{1}{2n^4} \text{ for any }\ell\neq m \text{ and } (a_1,a_2)\in\{0,1\}^2 \right]\leq O\left(\frac{1}{n^{\gamma -3}}\right).
  \end{align*}
  Using this  with~\eqref{eq:bound_phiv_phivtilde}, and taking a union bound,
  \begin{align*}
    \Pr\left[ \left|\phi_{\tilde{V},\bfx,\ell,m}(a_1,a_2)-\frac{1}{4} \right|>\frac{1}{n^4},\text{ for some }\ell,m\in[nR], a_1,a_2\in\{0,1\}\right]\leq O\left(\frac{1}{n}\right)
  \end{align*}
  provided that $\gamma> 15$.
By the nature of the expurgation process, we have 
\[
\phi_{\tilde{V},\bfx}^{(i)}(\tilde{c}^b) \leq \frac{1}{|f_i^{-1}(x_i)|}+\frac{1}{n^22^b}
\]
for every $i\in[n]$ and $\tilde{c}^b\in\{0,1\}^b$.
This implies that with probability $1-O(1/n)$, perturbation $\eta=\phi_{\tilde{V},\bfx}-\phi_{I,\bfx}$ satisfies all the properties we require---recall that the bit-pair marginals of $\phi_{I,\bfx}$ are uniform, see Property 5 after Definition~\ref{marcon}.
This completes the proof of Lemma~\ref{lemma:eta_good_existence}.\qed

\section{Proof of Proposition~\ref{lemma:exp_decay_pe_step4aa}}
\label{lapro3}  
Fix $\varepsilon>0$ so that $p+\varepsilon\leq 1/2$ and let $R=H(p+\varepsilon)$. From Proposition~\ref{lemma:NI_prob_error} there exists a rate-$R$ code ${\cal{C}}=( P_{\bo{C}|\bo{X}},\cG,f)$ that is private and that satisfies 
  $$\Pr[d_{H}(\hat{\bo{X}},\bo{X})\geq \delta n]\leq \frac{\text{poly}({\log n})}{\delta  n}$$
  for any $\delta>0$ (recall that $\hat{\bo{X}}$ is a function of ${\bo{C}}$). Expanding the left-side we have
\begin{align}
    \Pr[d_{H}(\hat{\bo{X}},\bo{X})\geq \delta n]= \mathbb{E}_\bo{X} \Pr_{\bo{C}|\bo{X}}[d_{H}(\hat{\bo{X}},\bo{X})\geq \delta n],
\end{align}
hence by Markov inequality we have
\begin{align}\label{mrk}
    \Pr_\bo{X}\bigg\{\bo{x}:\Pr_{\bo{C}|\bo{x}}[d_{H}(\hat{\bo{X}},\bo{x})\geq \delta n]\leq 1/k\bigg\}\geq 1-\frac{k\cdot \text{poly}({\log n})}{\delta n}
\end{align}
for any $\delta,k>0$.

  To establish the proposition, we use the code $\cal{C}$ in a concatenated fashion. Consider a length $N$ source sequence
  $$\boo{X}=\bo{X}(1),\bo{X}(2),\ldots,\bo{X}(k)$$
  composed of $k\geq 1$ consecutive blocks of length $n$.
   
Define a realization $\boo{x}$ of $\boo{X}$ to be $(\delta,k)$-\emph{distortion typical} if 
 $$\Pr_{\bo{C}|\bo{X}}[d_{H}(\hat{\bo{X}}(i),\bo{x}(i))\geq \delta n]\leq 1/k$$
   for all $i\in \{1,2,\ldots,k\}$. For any fixed $\delta,k>0$ this is a high probability event as $N$ gets large since by \eqref{mrk} we have 
   \begin{align}\label{tw3}
       \Pr[\boo{X}\:\text{is $(\delta,k)$-distortion typical}]\geq \left( 1-\frac{k^2 \cdot \text{poly}({\log N})}{\delta N}\right)^k.
   \end{align}
   
Now, to encode $$\boo{x}=\bo{x}(1),\bo{x}(2),\ldots,\bo{x}(k)$$ we encode each $\bo{x}(i)$ as follows. First order the codewords of $\cal{C}$ in the decreasing order of their induced Hamming distortions with the first block $\bo{x}(1)$, that is \begin{align*}
  d_H(\hat{\bo{x}}(\bo{c}_1(1)),\bo{x}(1))&\geq d_H(\hat{\bo{x}}(\bo{c}_2(1)),\bo{x}(1))\\
  &\geq \ldots d_H(\hat{\bo{x}}(\bo{c}_{2^{nR}}(1)),\bo{x}(1)),
  \end{align*}
  where $\bo{c}_j(1)$ denotes the codeword with the $j$-th largest distance from $\bo{x}(1)$. The first row of rectangles in Fig.~\ref{fig:scheme_step4a} represents the $2^{nR}$ possible codewords for the first block, from $\bo{c}_1(1)$ (for the left-most) to $\bo{c}_{2^{nR}}(1)$ (for the right-most). 
  The width of the $j$-th rectangle is equal to the probability that the encoder outputs this codeword given $\bo{x}(1)$, \ie, $$P_{\bo{C}|\bo{x}(1)}(\bo{c}_j(1)|\bo{x}(1)).$$
 
The codewords corresponding to the $i$-th block are depicted on the $i$-th row of Fig.~\ref{fig:scheme_step4a} (here $k=5$). They are initially arranged in decreasing order according to the distortion they induce with $\bo{x}_i$, and then circularly shifted to the right by $(i-1)/k$. 
Hence, the rectangles corresponding to codewords with the largest induced distortion with the $i$'th block start from position $(i-1)/k$ in the $i$-th row. For each row $i$, a codeword $\bo{c} $ such that $d_H(\hat{\bo{x}}(\bo{c}),\bo{x}(i))\geq \delta n$ is denoted in red, and otherwise it is in green. Observe that if $\boo{x}$ is distortion typical, any vertical line in Fig.~\ref{fig:scheme_step4a} intersects at most one red rectangle.

 \begin{figure}
    \begin{center}
    \includegraphics[width=0.4\textwidth]{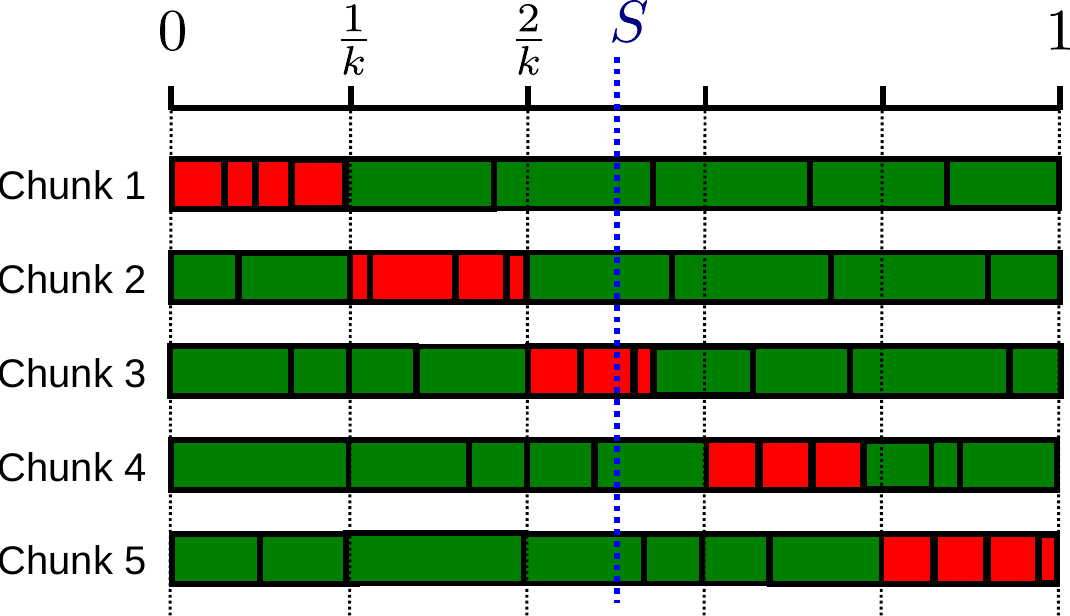}
    \caption{The coupled–concatenated scheme of Proposition~\ref{lemma:exp_decay_pe_step4aa} illustrated for a toy example with $k=5$.  
In the $i$th row, the red rectangles indicate the codewords whose Hamming distortion with respect to the $i$th source block $\bo{x}(i)$ exceeds $\delta n$.  For the realization of the random variable $S$ shown in this example, the encoder selects the sixth, fifth, second, eleventh, and eighth codewords from the respective lists of blocks $1$--$5$.
}\label{fig:scheme_step4a}
    \end{center}
  \end{figure}

  Now, given ${\boo{x}}$ we generate a codeword for each of its blocks as follows.  Referring to Fig.~\ref{fig:scheme_step4a}, let $S\in[0,1]$ be uniformly distributed over $[0,1]$, and independently drawn from ${\boo{x}}$. The codeword $\boo{C}$ for ${\boo{x}}$ is the concatenation of $k$ codewords 
  $$\boo{C}=\bo{C}_1,\ldots,\bo{C}_k$$
  where $\bo{C}_i$ is the codeword corresponding to the rectangle that includes $S$ in the $i$-th row---in the example of Fig.~\ref{fig:scheme_step4a}, block $\bo{x}(3)$ will be encoded with a ``large'' distortion codeword depicted in red, whereas $\bo{x}(2)$ will be encoded with a small distortion codeword. 

The reconstruction $\hat{\boo{X}}$ of ${\boo{x}}$ is obtained by decoding each block using the decoder of the baseline code. It should be emphasized here that the coupling variable $S$ is part of the encoding procedure only and not part of the decoding. Indeed, the codeword coupling does not affect each block's marginal codeword distribution. As a consequence, the concatenated code inherits the privacy property of the baseline code.

  Now for the distortion between ${\boo{x}}$ and  $\hat{\boo{X}}$. If ${\boo{x}}$ is $(\delta,k)$-distortion typical, the line at position $S$ will intersect at most one codeword yielding a distortion greater than $\delta n$ (and at most $n$). Therefore
  \begin{align}\label{tw1}
       d_{H}({\boo{x}},\hat{{\boo{X}}})\leq  n+k \delta n= N  (\delta+1/k) 
       \end{align}
  with probability one (over codeword distribution), for any $(\delta,k)$-distortion typical source sequence ${\boo{x}}$. 

Finally, since the integer $k$ can be chosen arbitrarily large and  $\delta>0$ arbitrarily small, from \eqref{tw3} and \eqref{tw1} we deduce that, for any fixed $\eta>0$, there exists a set of source sequences ${\cal{S}}_N$ with the following properties:
\begin{enumerate}
  \item Its probability satisfies  
\[
  \Pr[{\cal{S}}_N]\geq 1-\frac{ \text{poly}({\log N})}{N};
  \]
\item It admits a private rate-$R=H(p+\varepsilon)$ compressor with the property that, for every $x^{N}\in{\cal{S}}_N$, we have $$d_{H}(x^{N},\hat{X}^{N})\leq \eta N$$ with probability one with respect to the randomized encoding.\qed
\end{enumerate}

\section{Concluding remarks}\label{concrem}
We showed that lossless compression can accommodate private local decoding for memoryless sources. Several open problems remain. One would be to find a private, explicit, and non-asymptotic compression scheme. Another problem would be to investigate privacy and locality. We know that constant local decodability is achievable at any rate above entropy (for memoryless sources). Does this still hold if local decoding is subject to privacy? Note here the proposed scheme has locality $b=\Theta(\log n)$. Finally, an interesting extension to the present work would be to consider sources with memory.  


\bibliographystyle{IEEEtran}
\bibliography{codes_with_locality}

\appendix

\begin{lemma}\label{lemma:excess_fraction_bound}
Consider any positive integer $D$ and probability mass function $(p_1,\ldots,p_D)$ that satisfies
\[
\sum_{i=1}^D \left( p_i-\frac{1}{D} \right)^2\leq \frac{1}{\alpha D},
\]
for some $\alpha>0$. Then, for any $\varepsilon>0$,  we have 
\[
\sum_{i=1}^d \left[ p_i- (1+\varepsilon)\frac{1}{D} \right]^+\leq \frac{1+\varepsilon}{\alpha \varepsilon^2}
\]
where $[x]^+\defeq \max\{0,x\}$.
\end{lemma}
\begin{proof}[Proof of Lemma~\ref{lemma:excess_fraction_bound}]

Let us sort $\{1,2,\ldots,D\}$ in the decreasing order of $p_i$'s and let $\cD$ denote the sorted sequence. For $0<\rho<1$ such that $\rho D\in\bZ_{+}$, let $B(\rho)$ denote the first $ \rho D$ elements of $\cD$, equivalently, the $ \rho D$ largest probability values. Further, let 
\[
S(\rho)\defeq \sum_{i\in B(\rho)}p_i
\]
denote the total probability included in $B(\rho)$.

Using Jensen's inequality, 
\begin{align*}
   \rho D \frac{1}{ \rho D} \sum_{i\in B(\rho)}\left( p_i-\frac{1}{D}\right)^2&\geq  \rho D \left( \frac{S(\rho)}{\rho D}-\frac{1}{D} \right)^2  = \frac{\rho}{D}\left(\frac{S(\rho)}{\rho}-1\right)^2,
\end{align*}
and by choosing 
$$\rho=\frac{1}{D}\left\lfloor \frac{D}{\alpha\varepsilon^2}\right\rfloor \leq \frac{1}{\alpha\varepsilon^2},$$
we get
\begin{equation}
     \sum_{i\in B(\rho)}\left( p_i-\frac{1}{D}\right)^2 \geq \frac{1}{\alpha\varepsilon^2 D}\left(\frac{S(\rho)}{\rho}-1\right)^2.
\end{equation}
Using the hypothesis of the lemma, we have
\begin{equation*}
\frac{1}{\alpha\varepsilon^2 D}\left(\frac{S(\rho)}{\rho}-1\right)^2 \leq \frac{1}{\alpha D},
\end{equation*}
which yields
\begin{equation}\label{eq:excess_frac_bd_1}
\frac{S(\rho)}{\rho}\leq 1+\varepsilon,
\end{equation}
or, equivalently, using our choice of $\rho$
\begin{equation}\label{eq:excess_frac_bd_q}
S(\rho)\leq \frac{(1+\varepsilon)}{\alpha \varepsilon^2}.
\end{equation}

From the definition of $B(\rho)$, we have that
\[
S(\rho) \geq \rho D \min_{i\in B(\rho)}p_i,
\]
which when combined with \eqref{eq:excess_frac_bd_1} gives us
\[
\min_{i\in B(\rho)}p_i\leq (1+\varepsilon)\frac{1}{D}.
\]
This implies for any $i\notin B(\rho)$ we have
\[
p_i\leq (1+\varepsilon)\frac{1}{D}.
\]
Therefore,
\begin{align*}
    \sum_{i=1}^D \left[ p_i-(1+\varepsilon)\frac{1}{D} \right]^+ &\leq \sum_{i\in B(\rho)}p_i=S(\rho)\leq \frac{(1+\varepsilon)}{\alpha\varepsilon^2},
\end{align*}
which yields the desired result. \end{proof}

\end{document}